\documentclass[10pt,journal]{IEEEtran}
\usepackage{amsmath}
\usepackage{amsthm}
\usepackage{amsfonts}
\usepackage{graphicx}
\usepackage{latexsym}
\usepackage{amssymb}
\usepackage{color}
\usepackage{url}
\usepackage{colortbl}

\usepackage{slashbox}
\usepackage{xfrac}

\usepackage{caption}

\DeclareMathAlphabet{\mathbfsl}{OT1}{ppl}{b}{it} %{OT1}{cmr}{bx}{it}

\newcommand{\deff}{\mbox{$\stackrel{\rm def}{=}$}}

\makeatletter
\DeclareRobustCommand{\nsbinom}{\genfrac[]\z@{}}
\makeatother

\newcommand{\sbinomq}[2]{\nsbinom{{#1}}{{#2}}_{q}}
\newcommand{\sbinomtwo}[2]{\nsbinom{{#1}}{{#2}}_{2}}

\newcommand{\field}[1]{\mathbb{#1}}

\newcommand{\F}{\field{F}}

\newcommand{\dS}{\field{S}}

\newcommand{\cA}{{\cal A}}
\newcommand{\cB}{{\cal B}}
\newcommand{\cC}{{\cal C}}

\newcommand{\cG}{{\cal G}}

\newcommand{\C}{{\mathbb C}}

\newcommand{\linadd}{\kern1pt\mbox{\small$\boxplus$}\kern1pt}

\newtheorem{theorem}{Theorem}
\newtheorem{lemma}{Lemma}
\newtheorem{remark}{Remark}

\newtheorem{corollary}{Corollary}

\newtheorem{example}{Example}

\theoremstyle{definition}

\begin{document}

\bibliographystyle{plain}

\title{Grassmannian Codes with New\\Distance Measures for Network Coding}

\author{
{\sc Tuvi Etzion, \textit{Fellow, IEEE}}\thanks{Department of Computer Science, Technion,
Haifa 3200003, Israel, e-mail: {\tt etzion@cs.technion.ac.il}.
%Part of this research was supported by the  BSF-NSF grant 2016692.
} \hspace{1cm}
{\sc Hui Zhang}\thanks{Nanyang Technological University, Singapore, e-mail: {\tt huizhang@ntu.edu.sg}.
Part of the research was performed while
the author was with the Department of Computer Science, Technion,
Haifa 3200003, supported in part by a fellowship of the Israel Council of Higher
Education.~~~~~~~~~~~~~~~~~~~~~~~~~~~~~~~~~~~~~~~~~~~~~~~~~~~~   \newline
Parts of this work have been presented at the \emph{IEEE International Symposium on Information Theory 2018}, Vail, Colorado, U.S.A., June 2018.}
}

%\author{\IEEEauthorblockN{Tuvi Etzion, \textit{Fellow, IEEE}, Antonia Wachter-Zeh, \textit{Member, IEEE}}

\maketitle

\begin{abstract}
Grassmannian codes are known to be useful in error-correction for random network coding.
Recently, they were used to prove that vector network codes outperform scalar linear
network codes, on multicast networks, with respect to the alphabet size.
The multicast networks which were used for this purpose are generalized combination networks.
In both the scalar and the vector network coding solutions,
the subspace distance is used as the distance measure for the codes
which solve the network coding problem in the generalized combination networks.
In this work we show that the subspace distance can be replaced with two other
possible distance measures which generalize the subspace distance.
These two distance measures are shown to be equivalent under an orthogonal transformation.
It is proved that the Grassmannian codes with the new distance measures
generalize the Grassmannian codes with the
subspace distance and the subspace designs with the strength of the design.
Furthermore, optimal Grassmannian codes with the new distance measures
have minimal requirements for network coding solutions of some generalized combination networks.
The coding problems related to these two distance measures, especially
with respect to network coding, are discussed. Finally, by using these new concepts it is proved that codes in the
Hamming scheme form a subfamily of the Grassmannian codes.

\end{abstract}

\vspace{0.5cm}


\vspace{0.5cm}

%\noindent
%{\bf Mathematics Subject Classification}: 94B25, 05B40, 51E10  .

%\footnotetext[1] { This research was supported in part by the}

%%%%%%%%%%%%%%%%%%%%%%%%%%%%%%%%%%%%%%%%%%%%%%%%%%%%%%%%%%%%%%%%%%%%%%
%%%%%%%%%%%%%%%%%%%%%%%%%%%%%%%%%%%%%%%%%%%%%%%%%%%%%%%%%%%%%%%%%%%%%%
%%%%%%%%%%%%%%%%%%%%%%%%%%%%%%%%%%%%%%%%%%%%%%%%%%%%%%%%%%%%%%%%%%%%%%
%\newpage
\section{Introduction}

\looseness=-1
\IEEEPARstart{N}{etwork} coding has been attracting increasing attention in the last fifteen years.
The seminal work of Ahlswede, Cai, Li, and Yeung~\cite{ACLY00}
and Li, Yeung, and Cai~\cite{LYC03} introduced the basic concepts of network coding and
how network coding outperforms the well-known routing.
The class of networks which are mainly studied is the class of multicast networks and
these are also the target of this work. A \emph{multicast network} is a directed acyclic graph
with one source. The source has $h$ messages, which are scalars over a finite field $\F_q$.
The network has $N$ receivers, each one demands all the $h$ messages of the source
to be transmitted in one round of a network use.
An up-to-date survey on network coding for multicast networks can be found for example in~\cite{FrSo16}.
K\"{o}tter and M\'{e}dard~\cite{KoMe03} provided
an algebraic formulation for the network coding problem: for a given network, find coding
coefficients for each edge, whose starting vertex has in-degree greater
than one. These coding coefficients are multiplied by the symbols received at
the starting node of the edge and these products are added together. These coefficients should be chosen in a way
that each receiver can recover the $h$ messages from its received symbols
on its incoming edges.
This sequence of coding coefficients at each such edge is called the \emph{local coding vector}
and the edge is called a \emph{coding point}.
Such an assignment of coding coefficients for all such edges in the network is called a~\emph{solution} for the network
and the network is called \emph{solvable}. It is easy to verify that the information
on each edge is a linear combination of the $h$ messages.
The vector of length $h$ of these coefficients of this linear combination
is called the \emph{global coding vector}. From the global
coding vectors and the symbols on its incoming edges, the receiver should recover the $h$ messages,
by solving a set of $h$ linearly independent equations.
The coding coefficients defined in this way are scalars and the solution is a \emph{scalar linear solution}.
Ebrahimi and Fragouli~\cite{EbFr11}
have extended this algebraic approach to vector network coding. In the setting of vector
network coding, the messages of the source
are vectors of length~$\ell$ over $\F_q$ and the coding coefficients are $\ell \times \ell$ matrices
over~$\F_q$. A set of matrices, which have the role of the coefficients
of these vector messages, such that all the receivers
can recover their requested information, is called a \emph{vector solution}.
Also in the setting of vector network coding we distinguish between the local coding
vectors and the global coding vectors. There is a third type of network coding solution,
a scalar nonlinear network code. Again, in each coding point there is a function of the
symbols received at the starting node of the coding point. This function can be linear or
nonlinear. There is clearly a hierarchy, where a scalar linear solution can
be translated to a vector solution, and a vector solution can be translated to a scalar nonlinear solution.
%Finally, we should note that the setting with a source with
%$h$ messages is equivalent to a setting with $h$ sources, each one has a unique message.
%In our exposition, we assume the setting with a unique source.

The \emph{alphabet size} of the solution is an important parameter that directly influences
the complexity of the calculations at the network nodes and as a consequence the performance of the network.
A comparison between the required alphabet size for a scalar linear solution, a vector solution, and
a scalar nonlinear solution, of the same multicast network is an important problem.
It was proved in~\cite{EtWa15,EtWa16} that there are
multicast networks on which a vector network coding solution with vectors of length $\ell$
over $\F_q$ \emph{outperforms} any scalar linear network coding solution, i.e. the scalar
solution requires an alphabet of size $q_s$, where $q_s > q^\ell$.
The proof used a family of networks called the \emph{generalized combination networks},
where the combination networks were defined and used in~\cite{RiAh06}.

K\"{o}tter and Kschischang~\cite{KoKs08} introduced a framework for error-correction in
random network coding. They have shown that for this purpose the codewords (messages) are taken as
subspaces over a finite field $\F_q$. For this purpose they have
defined the \emph{subspace distance}. This approach was mainly applied on subspaces of the same dimension.
For given positive integers $n$ and~$k$, ${0 \leq k \leq n}$, the Grassmannian
$\cG_q(n,k)$ is the set of all subspaces of~$\F_q^n$ whose dimension is $k$.
It is well known that
$$ \begin{small}
| \cG_q (n,k) | = \sbinomq{n}{k}
\deff \frac{(q^n-1)(q^{n-1}-1) \cdots
(q^{n-k+1}-1)}{(q^k-1)(q^{k-1}-1) \cdots (q-1)}
\end{small}
$$
where $\sbinomq{n}{k}$ is the $q$-\emph{binomial coefficient} (known also
as the $q-$\emph{ary Gaussian coefficient}~\cite[pp. 325-332]{vLWi92}.
A code $\C \in \cG_q(n,k)$ is called a \emph{Grassmannian code} or a \emph{constant dimension code}.
For two subspaces $X, Y \in \cG_q(n,k)$ the subspace distance is reduced to the \emph{Grassmannian distance}
defined by

$$
d_G (X,Y) \deff k - \dim (X \cap Y) ~.
$$

Most of the research on Grassmannian codes motivated by~\cite{KoKs08} was in two directions -- finding the largest
codes with prescribed minimum Grassmannian distance and looking for designs based on subspaces.
To this end, the quantity $\cA_q (n,2d,k)$ was defined as the maximum size of a code
in~$\cG_q(n,k)$ with minimum Grassmannian distance $d$. There has been extensive work on Grassmannian codes
in the last ten years, e.g.~\cite{EtSi09,EtSi13,EtSt16,EtVa11,SKK08} and references therein.
A related concept is a \emph{subspace design}
or a \emph{block design} $t$-$(n,k,\lambda)_q$ which is a collection $\dS$ of
$k$-subspaces from $\cG_q(n,k)$ (called \emph{blocks}) such that each subspace of $\cG_q(n,t)$ is contained in
exactly $\lambda$ blocks of~$\dS$, where $t$ is called the \emph{strength} of the design.
In particular if $\lambda =1$ this subspace design is called a $q$-\emph{Steiner system}
and is denoted by $S_q(t,k,n)$. Note, that such a $q$-Steiner system is a Grassmannian code in $\cG_q(n,k)$
with minimum Grassmannian distance $k-t+1$. Such subspace designs were considered for example
in~\cite{BEOVW,BKL05,BKN15,Etz18,EtHo18,FLV14,KiLa15,KiPa15,MMY95,RaSi89,Suz90,Suz90a,Suz92,Tho87,Tho96}.

The goal of this work is to show that there is a tight connection between optimal Grassmannian codes
and network coding solutions for the generalized combination networks.
We will define two new dual distance measures on Grassmannian codes
which generalize the Grassmannian distance. We discuss the maximum sizes of Grassmannian
codes with the new distance measures and analyse these bounds
from a few different point of view. We explore the connection between these codes and related
generalized combination networks. Our exposition will derive some interesting properties of these codes
with respect to the traditional Grassmannian codes and some subspace designs.
We will show, using a few different approaches, that codes in the Hamming space
form a subfamily of the Grassmannian codes. Some other interesting connections to subspace designs
and codes in the Hamming scheme will be also explored.

The Grassmannian codes (constant dimension codes) are the $q$-analog of the constant weight codes,
where $q\text{-}$analogs replace concepts of subsets by concepts of subspaces when
problems on sets are transferred to problems on subspaces
over the finite field $\F_q$. For example, the size of a set is replaced by the dimension of a subspace,
the binomial coefficients are replaced by the Gaussian coefficients, etc.
The Grassmann space is the $q$-analog of the Johnson space and the subspace distance is the $q$-analog of
the Hamming distance. The new distance measures are $q$-analogs of related distances in the Johnson space.
The \emph{Johnson scheme} $J(n,w)$ consists of all $w\text{-}$subsets of an $n\text{-}$set (equivalent to binary words on length $n$ and weight~$w$).
The \emph{Johnson distance} $d_J(x,y)$ between two $w\text{-}$subsets $x$ and $y$ is
half of the related Hamming distance, i.e., $d_J(x,y) \triangleq |x \setminus y|$.

The rest of this paper is organized as follows. In Section~\ref{sec:gen_comb} we
present the combination network and its generalization which was defined in~\cite{EtWa15,EtWa16}.
We discuss the
family of codes which provide network coding solutions for these networks.
We will make a brief comparison between the related scalar coding solutions and vector coding solutions.
In Section~\ref{sec:gen_dist} we further consider this family of codes, define two dual distance measures
on these codes, and show how
these codes and the new distance measures defined on them generalize the conventional
Grassmannian codes with the Grassmannian distance. We show the connection of these codes to subspace designs.
We prove that for each such code
of the largest size, over $\F_q$, there exists a generalized combination network which is solved only by this code
(or another code with at least the same number of codewords and possibly more relaxed parameters).
Finally, we discuss which subfamily of these codes is useful for vector network coding.
In Section~\ref{sec:bounds} basic results on the upper bounds on sizes of these codes are presented.
At this point we note that some of the codes can have repeated codewords. Section~\ref{sec:bounds}
concentrates first on the case where there are no repeated codewords in the code.
It continues with upper bounds on sizes of codes where repeated codewords are considered.
In Section~\ref{sec:anal} we analyse the strength of our bounds and the implementation of the
codes on specific generalized combination networks.
In Section~\ref{sec:approach_Hamm} we discuss a few approaches to show how codes in the Hamming space
form a subfamily of codes in the Grassmann space.
We also discuss other connections of the newly defined distance measures and codes in the Hamming scheme.
Section~\ref{sec:conclude} provides a conclusion and some directions
for future research.

\section{Generalized Combination Networks}
\label{sec:gen_comb}

In this section we will define the generalized combination network which is a generalization
of the combination network~\cite{RiAh06}. This network defined
in~\cite{EtWa15,EtWa16} was used to prove that vector network coding outperforms scalar linear
network coding, in multicast networks, with respect to the alphabet size, using Grassmannian codes.

The \emph{generalized combination network} is called the
$\epsilon\text{-}$\emph{direct links} $k$-\emph{parallel links} $\mathcal{N}_{h,r,s}$ \emph{network},
in short the $(\epsilon,k)$-$\mathcal{N}_{h,r,s}$ network.
The network has three layers.
In the first layer there is a source
with $h$ messages. In the second layer there are $r$ nodes.
The source has $k$ parallel links to each node in the middle (second) layer.
From any $\alpha = \frac{s-\epsilon}{k}$ nodes
in the middle layer, there are links to one receiver in the third layer,
i.e. there are $\binom{r}{\alpha}$ receivers in the third layer.
From each one of these $\alpha$ nodes there are $k$ parallel links to the
related receiver in the third layer.
Additionally, from the source there are
$\epsilon$ direct parallel links to each one of the $\binom{r}{\alpha}$ receivers in the third layer.
Therefore, each receiver has $s = \alpha k+\epsilon$ incoming links.
The $(0,1)$-$\mathcal{N}_{h,r,s}$ network is the combination network defined in~\cite{RiAh06}.
This network has neither parallel links (between nodes) nor direct links (from the source
to the receiver).
We will assume some relations between the parameters $h$, $\alpha$, $\epsilon$, and $k$
such that the resulting network does not have a trivial solution or no solution.

\begin{theorem}[Theorem 8~\cite{EtWa15}]
The $(\epsilon,k)$-$\mathcal{N}_{h,r,\alpha k + \epsilon}$ network
has a trivial solution if $k + \epsilon \geq h$, and it has
no solution if $\alpha k + \epsilon < h$.
Otherwise, the network has a nontrivial solution.
\end{theorem}

Which network codes over~$\F_q$ solve the networks from this family of networks? The answer
to this natural question is quite simple. Since each receiver has
$\epsilon$ direct links from the source, it follows that the source can send
any required $\epsilon$-subspace of $\F_q^h$ to the receiver. Hence, the receiver must be able to obtain
an $(h-\epsilon)$-subspace of~$\F_q^h$ from the $\alpha$~middle layer nodes connected
to it. Each one of these $\alpha$ nodes in the middle layer can send to the receiver
a $k$-subspace of $\F_q^h$ (we can assume w.l.o.g. that each node of the middle
layer holds a $k$-subspace of $\F_q^h$.). Hence, a scalar linear
solution for the network exists if and only if the linear span of the $k$-subspaces
of each subset of $\alpha$ nodes in the middle layer is at least of dimension $h-\epsilon$.
Hence, a scalar linear solution for the network exists if and only if there exists a Grassmannian code $\C$ with $r$
$k$-subspaces of $\F_q^h$, such that each subset of $\alpha$ codewords ($k$-subspaces)
spans a subspace whose dimension is at least $h - \epsilon$. For this solution the coding coefficients
on a set of $s=\alpha k +\epsilon$ links ($\alpha$ sets with $k$-parallel links and one
set of $\epsilon$-direct links) are computed as follows.
The $k$ links from the source to the $i$-th node in the middle layer are associated with the $i$-th codeword
of $\C$. From this codeword any $k$ linearly independent vectors are chosen. The $h$ elements of $\F_q$
in each such vector of length $h$ are the $h$ coding coefficients on a related edge from the source to the $i$-th node
in the middle layer. The middle layer nodes transmit on its outgoing links the exact information
it receives from the source. Given a receiver $R$, the $\alpha k$ vectors on the edges entering $R$ (formed from
coding coefficients on the edges) span a subspace $X$ whose dimension is at least $h - \epsilon$. On the $\epsilon$
edges entering $R$ from the source, there are coding coefficients whose related vectors complete $X$ to $\F_q^h$.
These coding vectors form the scalar linear solution for the network.
A different analysis for the combination network, i.e., $\epsilon =0$ and $k =1$ was
given by Riis and Ahlswede~\cite{RiAh06}. For this we need to define an $(n,M,d)$ code over $\F_q$ to be a subset
of $M$ words of length~$n$ over~$\F_q$ with minimum Hamming distance $d$.
The functions on the edges related to the scalar nonlinear solution implied by the
following theorem will be explained (in the general context of the generalized combination
networks) in Section~\ref{sec:subfam_Hamm}, where they are relevant.

\begin{theorem}[Proposition 3~\cite{RiAh06}]
\label{thm:RiAh}
The $(0,1)$-$\mathcal{N}_{h,r,s}$ network is solvable over $\F_q$ if and only if there
exists an $(r,q^h,r-s+1)$ code over $\F_q$.
\end{theorem}

For a vector solution of the $(\epsilon,k)$-$\mathcal{N}_{h,r,\alpha k + \epsilon}$ network, the $h$ messages
are vectors of length~$\ell$ over~$\F_q$. Therefore, the total number of entries in the messages is $h \ell$
and the $h$ messages span an ${(h \ell)\text{-space}}$. Hence,
each receiver should obtain the space $\F_q^{h \ell}$ from its incoming edges. The source can send each receiver
an $(\epsilon \ell)$-subspace. This implies that each receiver should obtain a subspace whose dimension is at least
$(h - \epsilon)\ell$ from the related $\alpha$ middle layer nodes. On each set of $k$ parallel links
a node can send a $(k \ell)$-subspace of~$\F_q^{h \ell}$.
Hence, similarly to the scalar linear coding solution,
a vector solution for the network exists if and only if there exists a Grassmannian code with $r$
$(k \ell)$-subspaces of~$\F_q^{h \ell}$, such that each subset of $\alpha$ codewords ($(k \ell)$-subspaces)
spans a subspace whose dimension is at least $(h - \epsilon)\ell$.

%The generalized combination networks were used in~\cite{EtWa15,EtWa16} to show that vector network
%coding outperforms scalar linear network coding with respect to the alphabet size. To this end we
%have to find the largest codes which solve generalized combination networks with scalar linear network coding,
%and a large code which solve these networks with vector network coding. If the network has at least four messages
%then to show a gap in the alphabet size the $(\epsilon,k)$-$\mathcal{N}_{h,r,\alpha k + \epsilon}$ networks
%which were used have $\alpha =2$ and the required codes for vector network coding are the Grassmannian codes with
%the Grassmannian distance~\cite{EtWa15,EtWa16}. For three messages a larger $\alpha$ was required~\cite{EtWa15,EtWa16}.

\section{Covering/Multiple Grassmannian Codes}
\label{sec:gen_dist}

In this section we provide the formal definition for the codes required to solve the
generalized combination networks. We define two distance measures on these codes and
prove that Grassmannian codes and subspace designs,
are subfamilies of the related family of codes. We present some basic properties of these codes
and their connection to the network coding solutions for the generalized combination networks.

An $\alpha$-$(n,k,\delta)^c_q$ \emph{covering Grassmannian code} (code in short) $\C$ is a subset of $\cG_q(n,k)$
such that each subset of $\alpha$ codewords of $\C$ spans a subspace whose dimension is at
least $\delta +k$ in~$\F_q^n$.
The following theorem is easily verified.
\begin{theorem}
\label{thm:equ_dist}
$\C \in \cG_q(n,k)$ has minimum Grassmannian distance $\delta$ if and only if $\C$ is a\linebreak 2-$(n,k,\delta)^c_q$
code.
\end{theorem}
\begin{proof}
A code $\C \in \cG_q (n,k)$ has minimum Grassmannian distance $\delta$ if for each two distinct
codewords $X,Y \in \C$ we have $\delta \leq d_G (X,Y) = k - \dim (X \cap Y)$.
Since $\dim (X \cap Y) = \dim X + \dim Y - \dim (X \cup Y) = 2k - \dim (X \cup Y)$ it follows that
$$
\dim (X \cup Y) = 2k - \dim (X \cap Y) \geq k+ \delta ~,
$$
and hence $\C$ is a 2-$(n,k,\delta)^c_q$ code.

On the other hand if $\C$ is a 2-$(n,k,\delta)^c_q$ code then for each two distinct
codewords ${X,Y \in \C}$,
$$
k+ \delta \leq \dim  (X \cup Y) = \dim X + \dim Y - \dim (X \cap Y)
$$
$$
= 2k - \dim (X \cap Y)~,
$$
which implies that
$$
\delta \leq k - \dim (X \cap Y) = d_G (X,Y)~,
$$
i.e., $\C$ has minimum Grassmannian distance $\delta$.
\end{proof}

Theorem~\ref{thm:equ_dist} implies that the Grassmannian codes with the Grassmannian distance form a subfamily of the
$\alpha\text{-}(n,k,\delta)^c_q$ codes. It also implies that it is natural to define the quantity $\delta$ as
the minimum $\alpha\text{-}$\emph{Grassmannian distance} of the code and to define the quantity $k+\delta$
as the minimum $\alpha\text{-}$\emph{Grassmannian covering}, where the $2\text{-}$Grassmannian distance
is just the Grassmannian distance. In other words, the $\alpha\text{-}$\emph{Grassmannian covering}
of $\alpha$ subspaces in $\cG_q(n,k)$ is the dimension of the subspace which they span in~$\F_q^n$,
and the $\alpha\text{-}$\emph{Grassmannian distance} is the $\alpha\text{-}$\emph{Grassmannian covering}
minus $k$. A code $\C \in \cG_q(n,k)$ has minimum $\alpha\text{-}$\emph{Grassmannian distance} $\delta$ if
each subset of $\alpha$ codewords of $\C$ has $\alpha\text{-}$\emph{Grassmannian distance} at least $\delta$.
For such a code the minimum $\alpha\text{-}$\emph{Grassmannian covering} is $k+\delta$.
The quantity $\tilde{\cB}_q(n,k,\delta;\alpha)$ will denote the maximum size
of an ${\alpha\text{-}(n,k,\delta)^c_q}$ code. For the solution of the generalized combination network
one can use two identical codewords in such a code. But, from a coding theory point of view it is more
interesting to consider such codes in which there are no repeated codewords.
For this we define the quantity $\cB_q(n,k,\delta;\alpha)$ to be the maximum size
of an ${\alpha\text{-}(n,k,\delta)^c_q}$ code in which there are no repeated codewords.
Nevertheless, codes with repeated codewords are sometimes essential, at least for the network
coding solution, as is demonstrated in the next example.

\begin{example}
\label{ex:need_repeat}
Consider the $(1,1)$-$\mathcal{N}_{3,r,4}$ network in Figure~\ref{fig:PicNet}. In a scalar linear solution, a receiver
should get from its three incoming links a subspace of $\F_q^3$ whose dimension is
at least two. On each link a one-dimensional subspace of $\F_q^3$ is sent. Hence, the largest $r$ for
which the network is solvable over~$\F_q$ is $2 \sbinomq{3}{1}$. The solution consists of
a ${3\text{-}(3,1,1)^c_q}$ code which contains all one-dimensional subspaces of $\F_q^3$,
each one is contained twice in the code. This implies that $\tilde{\cB}_q(3,1,1;3) =2\sbinomq{3}{1}$,
while $\cB_q(3,1,1;3) =\sbinomq{3}{1}$.
\end{example}

\begin{figure}[t]
\centering
\includegraphics[width=0.88\columnwidth]{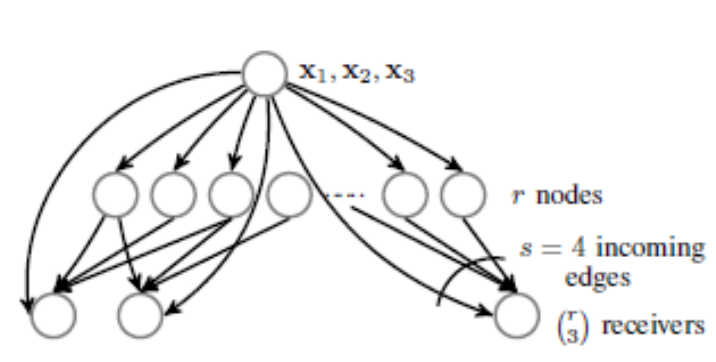}
\caption{the $(1,1)$-$\mathcal{N}_{3,r,4}$ network, each one of the $\binom{r}{3}$ receivers obtains its information
from three links of three distinct nodes in the middle layer and one direct link from the source,
where x$_1$, x$_2$, and x$_3$ are the three messages.}
\label{fig:PicNet}
\end{figure}

We continue in this section with the assumption that codewords can be repeated, since
the network coding solution can use repeated codewords.
In Section~\ref{sec:bounds} we will discuss bounds only for codes with non-repeated codewords,
and the difference when we allow repeated codewords.
From our discussion it can be readily verified
that if $\C$ is an $\alpha$-$(n,k,\delta)^c_q$ code which attains $\tilde{\cB}_q(n,k,\delta;\alpha)$ then
$\C$ solves the $(\epsilon,k)$-$\mathcal{N}_{n,r,\alpha k+\epsilon}$ network,
where $\epsilon \geq n-\delta-k$ and $r \leq \tilde{\cB}_q(n,k,\delta;\alpha)$.
Such a code implies the largest $r$ possible for such a network, given fixed $n$, $k$, $\alpha$, and $\epsilon$.
The largest $r$ means the maximum number of nodes that can be taken for the middle layer.
Clearly from our previous discussion, such a code has
parameters with the minimum requirements which are
necessary to solve such a network. These requirements are given
in the following theorem which generalizes Theorem~\ref{thm:RiAh}.

\begin{theorem}
\label{thm:gen_RiAh}
The $(\epsilon,k)$-$\mathcal{N}_{h,r,\alpha k +\epsilon}$ network is solvable over $\F_q$ if and only if there
exists an $\alpha$-$(h,k,h-k-\epsilon)^c_q$ code with $r$ codewords.
\end{theorem}

Thus, Theorem~\ref{thm:gen_RiAh} implies that each such covering Grassmannian code of the maximum size is exactly
what is required to solve a certain instance of the generalized combination networks.

The way in which the code solves the generalized combination network, as described,
is very natural when we consider the definition of the generalized combination network.
It implies the generalization for the Grassmannian distance, namely the $\alpha$-Grassmannian distance. Since
some might argue that this generalization is less natural from a point of view of a code definition,
we have also defined the $\alpha$-Grassmannian covering which yields a natural interpretation for
the $\alpha$-Grassmannian distance. Now, we will translate this definition into the
requirement from a packing point of view. For this purpose we will need to use the \emph{dual} subspace $V^\perp$
of a given subspace $V$ in~$\F_q^n$, and the orthogonal complement of a given code~$\C$.
For a code $\C$ in~$\cG_q(n,k)$ the
\emph{orthogonal complement} $\C^\perp$ is defined by
$$\C^\perp \deff \{ V^\perp ~:~ V \in \C \}~.$$
It is well-known~\cite[Lemma 13]{EtVa11} that the minimum subspace distance of~$\C$ and
the minimum subspace distance of~$\C^\perp$ are equal.
The following lemma is also well known.

\begin{lemma}[Lemma 12~\cite{EtVa11}]
\label{lem:two_perp}
For any two subspaces $U, V$ of $\F_q^n$ we have that $U^\perp \cap V^\perp =(U+V)^\perp$.
\end{lemma}
Clearly, by induction we have the following consequence from Lemma~\ref{lem:two_perp}.
\begin{corollary}
\label{cor:multi_perp}
For any given set of $\alpha$ subspaces $V_1,V_2,\ldots,V_\alpha$ of $\F_q^n$ we have
$$
\bigcap_{i=1}^\alpha V_i^\perp = \left( \sum_{i=1}^\alpha V_i \right)^\perp ~.
$$
\end{corollary}

Corollary~\ref{cor:multi_perp} induces a new definition of
a distance measure for the orthogonal complements of the Grassmannian codes
which solve the generalized combination~networks.
For a Grassmannian code
${\C \in \cG_q(n,k)}$, the minimum $\lambda$-\emph{multiple Grassmannian distance} is $k-t +1$, where $t$ is the smallest integer
such that each $t$-subspace of $\F_q^n$ is contained in at most $\lambda$
$k$-subspaces of $\C$. In the sequel it will be explained why this definition is a natural generalization
of the Grassmannian distance. Moreover, we will present it later in a more natural way,
which is similar to the definition of a subspace design.
These will be understood from the results given in the rest of this section.

\begin{theorem}
\label{thm:gen_dist}
If $\C \in \cG_q(n,k)$ is an $\alpha$-$(n,k,\delta)^c_q$ code then~$\C^\perp \in \cG_q(n,n-k)$
has minimum $(\alpha-1)$-\emph{multiple Grassmannian distance} $\delta$.
\end{theorem}
\begin{proof}
By the definition of an $\alpha$-$(n,k,\delta)^c_q$ code it follows that for each $\alpha$ subspaces
$V_1,V_2,\ldots,V_\alpha$ of~$\C$ we have that $\dim \bigl( \sum_{i=1}^\alpha V_i \bigr) \geq \delta +k$
and hence ${\dim \bigl( \sum_{i=1}^\alpha V_i \bigr)^\perp \leq n- \delta -k}$.
Therefore, by Corollary~\ref{cor:multi_perp} we have that $\dim \bigl( \bigcap_{i=1}^\alpha V_i^\perp \bigr) \leq n- \delta -k$.
This implies that each subspace of dimension $n-\delta -k+1$ of $\F_q^n$ can be contained in at most
$\alpha-1$ codewords of $\C^\perp$. Thus, since $\C^\perp \in \cG_q(n,n-k)$, it follows by definition that the minimum
$(\alpha-1)$-multiple Grassmannian distance of $\C^\perp$ is $(n-k)-(n-\delta-k+1)+1=\delta$.
\end{proof}

Theorem~\ref{thm:gen_dist} leads to a new definition for Grassmannian codes
(based on orthogonal complements of $\alpha$-$(n,k,\delta)^c_q$ codes).
A $t$-$(n,k,\lambda)^m_q$ \emph{multiple Grassmannian code} (code in short)~$\C$ is a subset of $\cG_q(n,k)$
such that each $t\text{-}$subspace of $\F_q^n$ is contained in at most $\lambda$ codewords of~$\C$.
Similarly, let $\tilde{\cA}_q (n,k,t;\lambda)$ denote the maximum size of a $t$-$(n,k,\lambda)^m_q$ code.
Let $\cA_q (n,k,t;\lambda)$ denote the maximum size of a $t$-$(n,k,\lambda)^m_q$ code, where there
are no repeated codewords.
Clearly, a subspace design $t\text{-}(n,k,\lambda)_q$ is a $t$-$(n,k,\lambda)^m_q$ code.
It should be noted that also in combinatorial designs repeated blocks are usually not
considered in the literature.
One can easily verify that
\begin{theorem}
\label{thm:equ_pack}
$\C \in \cG_q(n,k)$ has minimum Grassmannian distance ${k-t+1}$ if and only if $\C$ is a $t$-$(n,k,1)^m_q$ code.
\end{theorem}
\begin{proof}
If $\C$ has minimum Grassmannian distance ${k-t+1}$, then for each two distinct codewords
$X,Y \in \C$ we have $k-t+1 \leq d_G (X,Y) = k - \dim (X \cap Y)$. It implies that
$\dim (X \cap Y) \leq t-1$ and hence each $t\text{-}$subspace of $\F_q^n$ is contained in
at most one codeword of~$\C$, i.e., $\C$ is a $t$-$(n,k,1)^m_q$ code.

On the other hand, if $\C$ is a $t$-$(n,k,1)^m_q$ code, then any $t$-subspace of $\F_q^n$ is contained in at
most one codeword of~$\C$. It follows that for any $X,Y \in \C$ we have $\dim (X \cap Y) \leq t-1$.
Hence,
$$
d_G (X,Y) = k - \dim (X \cap Y) \geq k - (t-1) = k-t+1~,
$$
and therefore, $\C$ has minimum Grassmannian distance ${k-t+1}$.
\end{proof}

Clearly, by Theorem~\ref{thm:equ_pack} we have that the $\lambda$-multiple Grassmannian distance is a generalization
of the Grassmannian distance (which is a $1$-multiple Grassmannian distance).
The generalization implied by Theorem~\ref{thm:equ_pack} is for the packing interpretation of a $t$-$(n,k,1)^m_q$ code.
This is also a generalization of block design over $\F_q$ (a subspace design).
If each $t$-subspace is contained exactly once in a $t$-$(n,k,1)^m_q$ code $\C$, then $\C$ is
a $q$-Steiner system $S_q(t,k,n)$. If
each $t$-subspace is contained exactly $\lambda$ times in a $t$-$(n,k,\lambda)^m_q$ code $\C$, then $\C$ is
a $t$-$(n,k,\lambda)_q$ subspace design. Similarly to Theorem~\ref{thm:gen_dist} we have

\begin{theorem}
\label{thm:gen_dist2}
If $\C \in \cG_q(n,k)$ is a $t$-$(n,k,\lambda)^m_q$ code then~$\C^\perp$
has minimum $(\lambda+1)$-\emph{Grassmannian covering} ${n -t +1}$
and minimum $(\lambda+1)$-\emph{Grassmannian distance} ${k -t +1}$.
\end{theorem}
\begin{proof}
If $\C$ is a $t$-$(n,k,\lambda)^m_q$ code then for each $\lambda +1$
subspaces ${X_1,X_2,\ldots,X_{\lambda+1} \in \C}$ we have that
$$
\dim \bigcap_{i=1}^{\lambda+1} X_i < t ~.
$$
Therefore,
$$
t-1 \geq \dim \left( \left( \bigcap_{i=1}^{\lambda+1} X_i \right)^\perp \right)^\perp =  \dim \left( \sum_{i=1}^{\lambda+1} X_i^\perp \right)^\perp  ~,
$$
and hence,
$$
\dim \sum_{i=1}^{\lambda+1} X_i^\perp \geq n-t+1 ~.
$$
Thus, $\C^\perp \in \cG_q(n,n-k)$ has minimum $(\lambda+1)$-\emph{Grassmannian covering} ${n -t +1}$ and
minimum $(\lambda+1)$-\emph{Grassmannian distance} $(n-t+1)-(n-k)=k -t +1$.
\end{proof}
Combining Theorems~\ref{thm:gen_dist} and~\ref{thm:gen_dist2} yields the following results.
\begin{corollary}
\label{cor:equiv}
$~$
\begin{itemize}
\item[(1)]\label{cor:eq_perp1}
$\C \in \cG_q(n,k)$ is an $\alpha$-$(n,k,\delta)^c_q$ code if and only if $\C^\perp$
is an $(n-k-\delta+1)$-$(n,n-k,\alpha-1)^m_q$ code.
\item[(2)]\label{cor:equ_boundsR}
For any feasible $\delta$, $k$, $n$, and $\alpha$, $\tilde{\cB}_q (n,k,\delta;\alpha)=\tilde{\cA}_q(n,n-k,n-k-\delta+1;\alpha-1)$.
\item[(3)]\label{cor:equ_bounds1}
For any feasible $\delta$, $k$, $n$, and $\alpha$, $\cB_q (n,k,\delta;\alpha)=\cA_q(n,n-k,n-k-\delta+1;\alpha-1)$.
\item[(4)]\label{cor:eq_perp2}
$\C \in \cG_q(n,k)$ is a $t$-$(n,k,\lambda)^m_q$ code if and only if $\C^\perp$
is a $(\lambda+1)$-$(n,n-k,k-t+1)^c_q$ code.
\item[(5)]\label{cor:equ_bounds2}
For any feasible $t$, $k$, $n$, and $\lambda$, $\tilde{\cA}_q (n,k,t;\lambda)=\tilde{\cB}_q(n,n-k,k-t+1;\lambda+1)$.
\item[(6)]\label{cor:equ_bounds2}
For any feasible $t$, $k$, $n$, and $\lambda$, $\cA_q (n,k,t;\lambda)=\cB_q(n,n-k,k-t+1;\lambda+1)$.
\end{itemize}
\end{corollary}

Theorem~\ref{thm:gen_RiAh} provides a necessary and sufficient condition for the
requirements to solve the
$(\epsilon,k)$-$\mathcal{N}_{h,r,s}$ network with a scalar linear network code over $\F_q$.
Theorem~\ref{thm:gen_RiAh} is generalized for a solution with vector network coding whose
vectors have length~$\ell$, as follows.

\begin{theorem}
\label{thm:gen_RiAh_Vec}
The $(\epsilon,k)$-$\mathcal{N}_{h,r,\alpha k +\epsilon}$ network is solvable with vectors of length $\ell$
over $\F_q$ if and only if there
exists an $\alpha$-$(h\ell,k \ell ,h\ell -k\ell-\epsilon \ell)^c_q$ code with $r$ codewords.
\end{theorem}
\begin{proof}
Since there are $h$ messages, each one is a vector of length $\ell$, it follows that the source has
a subspace of dimension $h\ell$. The source can send to each receiver in the $\epsilon$ direct links
any subspace of dimension $\epsilon \ell$. Therefore, any $\alpha$ nodes in the middle layer should
send to the related receiver an $(h\ell -\epsilon \ell)$-subspace of $\F_q^{h\ell}$. The source sends
to each one of the $r$ nodes in the middle layer a $(k\ell)$-subspace and hence for a vector solution to the
$(\epsilon,k)$-$\mathcal{N}_{h,r,\alpha k +\epsilon}$ network, a Grassmannian code $\C \in \cG_q(h\ell,k\ell)$ of size $r$, with
minimum $\alpha$-Grassmannian covering $h\ell -\epsilon \ell$, i.e., minimum
$\alpha$-Grassmannian distance $h\ell -\epsilon \ell -k\ell$, is required. Hence, $\C$
is an $\alpha$-$(h\ell,k \ell ,h\ell -k\ell-\epsilon \ell)^c_q$ code with $r$ codewords.
\end{proof}

Can the scalar linear solution and the vector solution be compared only on the basis of
Theorems~\ref{thm:gen_RiAh} and~\ref{thm:gen_RiAh_Vec}?
Assume we are given the $(\epsilon,k)$-$\mathcal{N}_{h,r,\alpha k +\epsilon}$ network.
A scalar linear solution over $\F_{q^\ell}$ requires by Theorem~\ref{thm:gen_RiAh} an
$\alpha$-$(h,k,h-k-\epsilon)^c_{q^\ell}$ code with $r$ codewords. The related vector network coding with vectors
of length $\ell$ over $\F_q$ requires an $\alpha$-$(h \ell,k \ell,h\ell-k\ell-\epsilon\ell)^c_q$ code with $r$ codewords.
One can construct the vector network code from the scalar linear network code by using companion matrices
and their powers~\cite{EbFr11,EtWa15,EtWa16}. But, the important question is whether we can find a Grassmannian
code for the vector network coding larger than the largest one for scalar linear network coding.
Some examples of such codes are given in~\cite{EtWa15,EtWa16} and other codes are a subject for further research.

\section{Upper Bounds on the Sizes of Codes}
\label{sec:bounds}

In this section we will present some upper bounds on the sizes of codes.
Upper bounds on the sizes of codes with no repeated codewords are considered first
and later a short discussion is given for codes with repeated codewords.

\subsection{Bounds on Sizes of Codes with no Repeated Codewords}
\label{sec:bounds_no_repeat}

Clearly, there is a huge ground for research since the parameters
of the codes are in a very large range and our knowledge is very limited.
We will give some ideas and some insight about the difficulty
of obtaining new bounds and especially the exact size of optimal codes.
The bounds are on $\cA_q (n,k,t;\lambda)$ and on $\cB_q(n,k,\delta;\alpha)$ and clearly
by Corollary~\ref{cor:equiv}(3) and~\ref{cor:equiv}(6), bounds are required
only on one of them since they are equivalent (when the related parameters are taken).
There is a duality between the two types of codes which were considered with the two dual distance measures.
Hence, it is sometimes simpler and more convenient to analyse or construct a large code
with one of the two distance measures. As mentioned before, the case of $\lambda=1$ was considered
in the last ten years, and the following simple equality
reduced the range for the search of such bounds.
\begin{theorem}[Lemma 13~\cite{EtVa11}]
\label{thm:lambda=1}
If $n$, $k$, $t$, are positive integers such that $1 \leq t <k <n$, then
$\cA_q(n,k,t;1) = \cA_q(n,n-k,n-2k+t;1)$.
\end{theorem}
%\begin{proof}
%Let $\C$ be a $t$-$(n,k,1)^m_q$ code with $M=\cA_q(n,k,t;1)$ codewords. The code
%has Grassmannian distance $k-t+1$. The orthogonal complement code $\C^\perp$ is a code in $\cG_q(n,n-k)$
%with the same minimum Grassmannian distance $k-t+1$. Hence, by Theorem~\ref{thm:equ_pack} it is an
%$(n-2k+t)$-$(n,n-k,1)^m_q$ code and $\cA_q(n,k,t;1) \leq \cA_q(n,n-k,n-2k+t;1)$.
%
%Now, let $\C$ be an $(n-2k+t)$-$(n,n-k,1)^m_q$ code with $M=\cA_q(n,n-k,n-2k+t;1)$ codewords. The code
%has Grassmannian distance $(n-k)-(n-2k+t)+1=k-t+1$. The orthogonal complement code $\C^\perp$ is a code in $\cG_q(n,k)$
%with the same minimum Grassmannian distance $k-t+1$. Hence, by Theorem~\ref{thm:equ_pack} it is a
%$t$-$(n,k,1)^m_q$ code and $\cA_q(n,k,t;1) \geq \cA_q(n,n-k,n-2k+t;1)$.
%
%Thus, $\cA_q(n,k,t;1) = \cA_q(n,n-k,n-2k+t;1)$.
%\end{proof}

Theorem~\ref{thm:lambda=1} implies that if $\lambda=1$ then it is enough to find
$\cA_q(n,k,t;1)$ for $k \leq n-k$. This is not the case when $\lambda >1$, where
an analysis will be considered in Section~\ref{sec:anal}. In this analysis
for $\lambda >1$ we will consider the case where $n \geq 2k$ as well as the case
where $n < 2k$.
In this section we will consider only the basic upper bounds.

We start by considering the duality between the two distance measures
and related different simple approaches to obtain
bounds on the maximum sizes of codes. To simplify and emphasize the
properties on which the bounds are analysed we summarize them. Our starting point will be
an $\alpha$-$(n,k,\delta)^c_q$ code.
\begin{itemize}
\item[(c.1)] In an $\alpha$-$(n,k,\delta)^c_q$ code $\C$, each subset of $\alpha$ codewords ($k$-subspaces) spans a
subspace of~$\F_q^n$ whose dimension is at least $\delta +k$.

\item[(c.2)] Each $(\delta +k-1)$-subspace of $\F_q^n$ contains
at most $\alpha -1$ codewords of $\C$ (by (c.1)).

\item[(c.3)] $\C^\perp$ is an ${(n-k-\delta+1)}$-${(n,n-k,\alpha -1)^m_q}$ code (see Corollary~\ref{cor:eq_perp1}(1)).
In such a code each $(n-\delta -k+1)$-subspace of $\F_q^n$ is contained in at most $\alpha-1$ codewords.

\item[(c.4)] Any $\alpha$ codewords from $\C^\perp$ intersect in a subspace whose dimension is at most $n-\delta-k$ (by (c.3)).
\end{itemize}
Bounds on the maximum size of related codes
can be obtained based on any one of these four observations and properties. Each one of these four
properties can give another direction to obtain related bounds.

If our starting point is a $t$-$(n,k,\lambda)^m_q$ code then the four dual properties are as follows.
\begin{itemize}
\item[(m.1)] In a $t$-$(n,k,\lambda)^m_q$ code $\C$, each $t$-subspace of $\F_q^n$ is contained in
at most $\lambda$ codewords.

\item[(m.2)] Any $\lambda +1$ codewords of $\C$ intersect in a subspace whose dimension
is at most $t-1$ (by (m.1)).

\item[(m.3)] $\C^\perp$ is a $(\lambda+1)$-$(n,n-k,k-t+1)^c_q$ code (see Corollary~\ref{cor:eq_perp2}(4)).
In such a code each subset of $\lambda+1$ codewords ($(n-k)$-subspaces) of $\F_q^n$ spans a subspace whose dimension is at least $(n-t+1)$.

\item[(m.4)] Each $(n-t)$-subspace of $\F_q^n$ contains at most $\lambda$ codewords of $\C^\perp$ (by (m.3)).
\end{itemize}

The classic bounds for the cases $\lambda =1$ (or $\alpha =2$, respectively) for a $t$-$(n,k,\lambda)^m_q$ code
(or an $\alpha$-$(n,k,\delta)^c_q$ code, respectively) can be easily generalized
for larger $\lambda$ ($\alpha$, respectively), where the simplest ones are the packing bound
and the Johnson bounds~\cite{EtVa11}. It might be easier to generalize these bounds when we consider
$t$-$(n,k,\lambda)^m_q$ codes and two proofs based on two of the four given properties can be given.
The following well-known lemma will be used frequently in our results.

\begin{lemma}
\label{lem:t_k_complete}
A $t$-subspace of $\F_q^n$ is contained in $\sbinomq{n-t}{k-t}$ distinct $k$-subspaces
of $\F_q^n$.
\end{lemma}
\begin{remark}
Note, that by Lemma~\ref{lem:t_k_complete} we have that if $\lambda = \sbinomq{n-t}{k-t}$ then
$\cA_q (n,k,t;\lambda) = \sbinomq{n}{k}$. Hence, the largest $\lambda$ that we should consider
is $\sbinomq{n-t}{k-t}$.
\end{remark}
The first few results are the $q$-analog of the packing bound.
\begin{theorem}
\label{thm:sphere}
If $n$, $k$, $t$, and $\lambda$ are positive integers such that $1 \leq t <k <n$ and $1 \leq \lambda \leq \sbinomq{n-t}{k-t}$,
then
$$
\cA_q (n,k,t;\lambda) \leq \left\lfloor \lambda \frac{\sbinomq{n}{t}}{\sbinomq{k}{t}}  \right\rfloor ~.
$$
\end{theorem}
\begin{proof}
Let $\C$ be a $t$-$(n,k,\lambda)^m_q$ code. There are $\sbinomq{n}{t}$ distinct $t$-subspaces
of $\F_q^n$ and hence
by (m.1) there are at most $\lambda \sbinomq{n}{t}$ pairs $(X,Y)$, where
$X$ is a $t$-subspace and $Y$ is a codeword in~$\C$ which contains $X$.
The number of $t$-subspaces in a codeword (a $k$-subspace) is $\sbinomq{k}{t}$,
and hence the theorem follows.
\end{proof}
\begin{remark}
Assume that $\alpha = \sbinomq{k+\delta-1}{k} +1$ and we consider $\C$ to be the Grassmannian
code which contains all the $\sbinomq{n}{k}$ subspaces of $\cG_q(n,k)$. If we take all the
$k$-subspaces in a $(k+\delta-1)$-subspace $X$ of $\F_q^n$ then they span only $X$, but with one
more subspace a $(k+\delta)$-subspace will be spanned.
Hence, the largest $\alpha$ that we should consider
is $\sbinomq{k+\delta-1}{k} +1$.
\end{remark}
\begin{corollary}
\label{cor:sphere}
If $n$, $k$, $\delta$, and $\alpha$ are positive integers such that $1 <k <n$,
$1 \leq \delta \leq n-k$ and $2 \leq \alpha \leq \sbinomq{k+\delta-1}{k}+1$, then
$$
\cB_q (n,k,\delta;\alpha) \leq \left\lfloor (\alpha -1) \frac{\sbinomq{n}{\delta +k-1}}{\sbinomq{n-k}{\delta -1}}  \right\rfloor ~.
$$
\end{corollary}
%\begin{proof}
%Let $\C$ be an $\alpha$-$(n,k,\delta)^c_q$ code. There are $\sbinomq{n}{\delta+k-1}$ distinct
%$(\delta +k-1)$-subspaces of $\F_q^n$ and hence
%by (c.2) there are at most $(\alpha -1) \sbinomq{n}{\delta +k-1}$ pairs $(X,Y)$, where
%$X$ is a $(\delta +k-1)$-subspace and $Y$ is a codeword (a $k$-subspace) of $\C$ which is contained in $X$.
%By Lemma~\ref{lem:t_k_complete}, a codeword in $\C$, which is a $k$-subspace,
%is contained in $\sbinomq{n-k}{\delta -1}$ distinct $(\delta +k-1)$-subspaces
%of~$\F_q^n$ and hence the claim in the corollary follows.
%\end{proof}
The rest of this section is devoted to the $q$-analog of the Johnson bounds~\cite{EtVa11,XiFu09}.
We should remark that there is an improvement of the Johnson bound for
Grassmannian codes~\cite{KiKu17}. The same improvement can be applied also
for Grassmannian codes with the new distance measures~\cite{EKOO18}.
We omit this improvement since we concentrate in this paper only on the basic bounds.
Also the improvement is for some parameters, it is not dramatic improvement, and far from being simple.
\begin{theorem}
\label{thm:Johnson1}
If $n$, $k$, $t$, and $\lambda$ are positive integers such that $1 \leq t <k <n$ and $1 \leq \lambda \leq \sbinomq{n-t}{k-t}$,
then
$$
\cA_q (n,k,t;\lambda) \leq \left\lfloor \frac{q^n -1}{q^k-1} \cA_q (n-1,k-1,t-1;\lambda)  \right\rfloor ~.
$$
\end{theorem}
\begin{proof}
Let $\C$ be a $t$-$(n,k,\lambda)^m_q$ code for which $|\C|=M=\cA_q (n,k,t;\lambda)$.
Each codeword of $\C$ contains $\frac{q^k-1}{q-1}$ one-dimensional subspaces of $\F_q^n$.
Since the number of one-dimensional subspaces of $\F_q^n$ is $\frac{q^n-1}{q-1}$, it follows
that there exists a one-dimensional subspace $X$ of $\F_q^n$ which is contained in at least
$\frac{q^k-1}{q^n-1} M$ codewords of $\C$.
Let $Y$ be an $(n-1)$-subspace of $\F_q^n$ such that
$\F_q^n =Y+X$. Define $\C' \deff \{ C \cap Y ~:~ C \in \C,~ X \subset C \}$. Clearly, $\C'$ is
a subspace code with $(k-1)$-subspaces of $Y$. Note, that for two distinct codewords $C_1,C_2 \in \C$
such that $X \subset C_1$ and $X \subset C_2$ we have $C'_1 = C_1 \cap Y \neq C_2 \cap Y =C'_2$ since
$C_1 = C'_1 +X$ and $C_2 = C'_2 +X$. Therefore, the size of~$\C'$ is at least $\frac{q^k-1}{q^n-1} M$.

Let $Z'$ be a $(t-1)$-subspace of $Y$.
$Z =Z'+X$ is a $t$-subspace of $\F_q^n$ and hence it is a subspace of at most $\lambda$ codewords
of $\C$. Let $C_1 , C_2 ,\ldots , C_s$, $s \leq \lambda$, be the only $s$ distinct codewords of $\C$ which contain $Z$.
Since $Z$ contains $X$ it follows that $C'_i = C_i \cap Y$, $1 \leq i \leq s$, are distinct codewords in $\C'$
which contain $Z'$ (note that $Z' = Z \cap Y$). If there exists another codeword $C' \in \C'$ such that
$Z' \subset C'$ then $C = C' +X$ is a codeword of~$\C$ which contains $Z$, a contradiction to the
fact that only $C_1 , C_2 ,\ldots , C_s$ are the codewords of $\C$ which contain $Z$. Hence, each
$(t-1)$-subspace of $Y$ is contained in at most $s \leq \lambda$ codewords of $\C'$.

Therefore, $\C'$ is a $(t-1)$-$(n-1,k-1,\lambda)^m_q$ code whose size is at least $\frac{q^k-1}{q^n-1} M$.
Hence, $\frac{q^k-1}{q^n-1} \cA_q (n,k,t;\lambda) \leq |\C'| \leq \cA_q (n-1,k-1,t-1;\lambda)$, and thus,
$$
\cA_q (n,k,t;\lambda) \leq \left\lfloor \frac{q^n -1}{q^k-1} \cA_q (n-1,k-1,t-1;\lambda)  \right\rfloor ~.
$$
\end{proof}
\begin{corollary}
\label{cor:Johnson1}
If $n$, $k$, $\delta$, and $\alpha$ are positive integers such that $1 <k <n$,
$1 \leq \delta \leq n-k$ and $2 \leq \alpha \leq \sbinomq{k+\delta-1}{k}+1$, then
$$
\cB_q (n,k,\delta;\alpha) \leq \left\lfloor \frac{q^n -1}{q^{n-k}-1} \cB_q (n-1,k,\delta;\alpha)  \right\rfloor ~.
$$
\end{corollary}
%\begin{proof}
%By Corollay~\ref{cor:equ_bounds1}, Theorem~\ref{thm:Johnson1}, and Corollary~\ref{cor:equ_bounds2}
%we have the following sequence of equalities and inequalities,
%$$
%\cB_q (n,k,\delta;\alpha) = \cA_q (n,n-k,n-k-\delta+1;\alpha -1)~~\text{by~Corollary~\ref{cor:equ_bounds1}}
%$$
%$$
%\leq \left\lfloor \frac{q^n -1}{q^{n-k}-1} \cA_q (n-1,n-k-1,n-k-\delta;\alpha -1)  \right\rfloor~~\text{by~Theorem~\ref{thm:Johnson1}}
%$$
%$$
%\leq \left\lfloor \frac{q^n -1}{q^{n-k}-1} \cB_q (n-1,k,\delta;\alpha)  \right\rfloor~~\text{by~Corollary~\ref{cor:equ_bounds2}}.
%$$
%\end{proof}
\begin{theorem}
\label{thm:Johnson2}
If $n$, $k$, $t$, and $\lambda$ are positive integers such that $1 \leq t <k <n$ and $1 \leq \lambda \leq \sbinomq{n-1-t}{k-t}$,
then
$$
\cA_q (n,k,t;\lambda) \leq \left\lfloor \frac{q^n -1}{q^{n-k}-1} \cA_q (n-1,k,t;\lambda)  \right\rfloor ~.
$$
\end{theorem}
\begin{proof}
Let $\C$ be a $t$-$(n,k,\lambda)^m_q$ code for which $|\C|=M=\cA_q (n,k,t;\lambda)$.
For each $(n-1)$-subspace $Y$ let $\C_Y \deff \{ C ~:~ C \in \C, ~ C \subset Y \}$.
Clearly, $\C_Y$ is a $t\text{-}(n-1,k,\lambda)^m_q$ code. By Lemma~\ref{lem:t_k_complete}, each $k$-subspace of~$\F_q^n$ is contained
in $\frac{q^{n-k}-1}{q-1}$ subspaces of $\cG_q(n,n-1)$. There exist $\frac{q^n-1}{q-1}$ subspaces
in $\cG_q(n,n-1)$ and hence there exists one $(n-1)$-subspace $Z$ such that
$$
|\C_Z| \geq \frac{q^{n-k}-1}{q^n-1} M~.
$$
Since also $\C_Z$ is a $t$-$(n-1,k,\lambda)^m_q$ code (see the definition of $\C_Y$),
it follows that $\frac{q^{n-k}-1}{q^n-1} \cA_q (n,k,t;\lambda) \leq |\C_Z| \leq \cA_q (n-1,k,t;\lambda)$, and thus,
$$
\cA_q (n,k,t;\lambda) \leq \left\lfloor \frac{q^n -1}{q^{n-k}-1} \cA_q (n-1,k,t;\lambda)  \right\rfloor ~.
$$
\end{proof}
\begin{corollary}
\label{cor:Johnson2}
If $n$, $k$, $\delta$, and $\alpha$ are positive integers such that $1 <k <n$,
$1 \leq \delta \leq n-k$ and $2 \leq \alpha \leq \sbinomq{k+\delta-2}{k-1}+1$, then
$$
\cB_q (n,k,\delta;\alpha) \leq \left\lfloor \frac{q^n -1}{q^k-1} \cB_q (n-1,k-1,\delta;\alpha)  \right\rfloor ~.
$$
\end{corollary}
%\begin{proof}
%By Corollary~\ref{cor:equ_bounds1}, Theorem~\ref{thm:Johnson2}, and Corollary~\ref{cor:equ_bounds2}
%we have the following sequence of equalities and inequalities,
%$$
%\cB_q (n,k,\delta;\alpha) = \cA_q (n,n-k,n-k-\delta+1;\alpha -1)~~\text{by~Theorem~\ref{thm:Johnson1}}
%$$
%$$
%\leq \left\lfloor \frac{q^n -1}{q^k-1} \cA_q (n-1,n-k,n-k-\delta+1;\alpha -1)  \right\rfloor~~\text{by~Theorem~\ref{thm:Johnson2}}
%$$
%$$
%\leq \left\lfloor \frac{q^n -1}{q^k-1} \cB_q (n-1,k-1,\delta;\alpha)  \right\rfloor~~\text{by~Corollary~\ref{cor:equ_bounds2}}.
%$$
%\end{proof}

We would also like to remind and to mention that some bounds can be obtained from known results
on arcs and caps in projective geometry~\cite{HiSt01}. Discussion on these is
given in~\cite{EKOO18}.

\subsection{Bounds on Sizes of Codes with Repeated Codewords}
\label{sec:bounds_repeat}

In this subsection we will discuss the differences between codes with repeated codewords on those
with no repeated codewords. This will be done with respect to the sizes of codes.
%It is important to understand that from coding theory (or block design) point of view we would
%like to find optimal codes with no repeated codewords. But, from the point of view of the
%network coding solutions for the generalized combination network, repeated codewords are allowed.
We note that usually the difference in the size of the codes is minor and in most cases this difference
can be ignored. Some cases where this difference is not minor will be considered in this subsection.
%\begin{theorem}
%\label{thm:sphereR}
%If $n$, $k$, $t$, and $\lambda$ are positive integers such that $1 \leq t <k <n$,
%then
%$$
%\tilde{\cA}_q (n,k,t;\lambda) \leq \left\lfloor \lambda \frac{\sbinomq{n}{t}}{\sbinomq{k}{t}}  \right\rfloor ~.
%$$
%\end{theorem}
%\begin{corollary}
%\label{cor:sphereR}
%If $n$, $k$, $\delta$, and $\alpha$ are positive integers such that $1 <k <n$, then
%$$
%\tilde{\cB}_q (n,k,\delta;\alpha) \leq \left\lfloor (\alpha -1) \frac{\sbinomq{n}{\delta +k-1}}{\sbinomq{n-k}{\delta -1}}  \right\rfloor ~.
%$$
%\end{corollary}
%The proofs of Theorem~\ref{thm:sphereR} and Corollary~\ref{cor:sphereR} are the same as the ones
Packing bounds for codes with repeated codewords are obtained exactly as in
Theorem~\ref{thm:sphere} and Corollary~\ref{cor:sphere}, respectively. Similar results can be obtained for the Johnson bounds.
We will omit these repetitions.
Although, the bounds seem to be the same there might be a small difference in the sizes of the codes,
as we have the following trivial results.
\begin{theorem}
\label{thm:diff_bound1}
If $n$, $k$, $t$, and $\lambda$ are positive integers such that $1 \leq t <k <n$ and $1 \leq \lambda \leq \sbinomq{n-t}{k-t}$,
then
$$
\tilde{\cA}_q (n,k,t;\lambda) \geq \cA_q (n,k,t;\lambda)~.
$$
\end{theorem}
\begin{corollary}
\label{cor:diff_bound2}
If $n$, $k$, $\delta$, and $\alpha$ are positive integers such that $1 <k <n$,
$1 \leq \delta \leq n-k$ and $2 \leq \alpha \leq \sbinomq{k+\delta-1}{k}+1$, then
$$
\tilde{\cB}_q (n,k,\delta;\alpha) \geq \cB_q (n,k,\delta;\alpha)~.
$$
\end{corollary}

Finally, since we allow repeated codewords we have the following trivial lower bound.
\begin{theorem}
\label{thm:multiple_gamma}
If $n$, $k$, $t$, $\lambda$, and $\lambda'$ are positive integers such that $1 \leq t <k <n$,
then
$$
\tilde{\cA}_q (n,k,t;\lambda'\lambda) \geq \lambda' \tilde{\cA}_q (n,k,t;\lambda)~.
$$
\end{theorem}

One can immediately infer from Theorem~\ref{thm:multiple_gamma} that there are many examples like
Example~\ref{ex:need_repeat} for which $\tilde{\cB}_q (n,k,\delta;\alpha) > \cB_q (n,k,\delta;\alpha)$,
e.g., $(\alpha -1) \sbinomq{n}{k} = \tilde{\cB}_q (n,k,1;\alpha) > \cB_q (n,k,1;\alpha)= \sbinomq{n}{k}$.
To end this section we remind the reader that for the solution of the generalized combination
network repeated codewords are allowed. On the other hand, usually in coding theory and block design,
repeated codewords are not allowed, or more precisely, the related code or design with repeated codeword
is considered to be not simple and less interesting (in the literature bounds are considered only for codes with no
repeated codewords).

\section{Analysis of the Related Codes}
\label{sec:anal}

How good are the upper bounds given in Section~\ref{sec:bounds}? Can the bound of Theorem~\ref{thm:sphere} be attained? If $\lambda =1$ then
in view of Theorem~\ref{thm:lambda=1} we have to consider $\cA_q(n,k,t;1)$ only for $k \leq n-k$, where good constructions
and asymptotic bounds are known~\cite{BlEt11,EtSi13}. If $\lambda >1$ then the situation is quite different.
In this section we present a short analysis of the lower bounds compared to the upper bounds.
We do this analysis from a few point of view which are described in related subsections.
We do not intend to give any specific construction as this is a topic for another
research work (see~\cite{EKOO18}). Note, that most constructions for $\lambda=1$ can be generalized for
larger $\lambda$ if $k \leq n-k$. As for $\lambda=1$, also for larger $\lambda$ it is not difficult
to design constructions based on projective geometry~\cite{EKOO18},
on Ferrers diagrams~\cite{EGRW16,EtSi09,SiTr15}, and on rank-metric codes~\cite{SKK08}.

We also take into account
the generalized combination networks solved by some of the given codes. Since each $t$-subspace is contained in
$\sbinomq{n-t}{k-t}$ different $k$-subspaces of~$\F_q^n$ we have that
\begin{theorem}
\label{thm:max_lambda}
If $n$, $k$, and $t$ are positive integers such that $1 \leq t <k <n$, then
$$
\cA_q \left( n,k,t;\sbinomq{n-t}{k-t} \right) = \sbinomq{n}{k}~.
$$
\end{theorem}

In view of the equality that $\cA_q (n,k,t;1) = \cA_q (n,n-k,n-2k+t;1)$ (see Theorem~\ref{thm:lambda=1})
and Theorem~\ref{thm:max_lambda}, it is interesting to have a hierarchy of the values
$\cA_q (n,k,t;\lambda)$, where $\lambda = 1,2,\ldots , \sbinomq{n-t}{k-t}$.
It should be noted that similarly to the claim in Theorem~\ref{thm:lambda=1} we have for $2k \leq n$,
\begin{equation}
\label{eq:comp_design}
\cA_q (n,k,t;\lambda) \leq \cA_q (n,n-k,n-2k+t;\lambda)
\end{equation}
by considering orthogonal complement codes. But, except for $\lambda=1$ such an
inequality in (\ref{eq:comp_design}) is unlikely to be equality. For example, we have for $2k \leq n$ that
$$
\cA_q \left(n,k,t;\sbinomq{n-t}{k-t}\right) = \sbinomq{n}{k}~,
$$
while for the dual dimension we have
$$
\cA_q \left( n,n-k,n-2k+t;\sbinomq{2k-t}{k-t} \right) = \sbinomq{n}{k}~.
$$
When $2k < n$ we have
$$
\sbinomq{2k-t}{k-t} < \sbinomq{n-t}{k-t}~,
$$
which implies that
$$
\cA_q \left(n,k,t;\sbinomq{2k-t}{k-t} \right)
$$
$$
< \cA_q \left(n,n-k,n-2k+t;\sbinomq{2k-t}{k-t}\right) = \sbinomq{n}{k}~.
$$

We have proved a duality between the two distance measures and their codes, but this is not
the only duality that can be obtained.
Let $\cC_q (n,k,t;\lambda)$ be the minimum size of a code (with distinct codewords) in $\cG_q(n,k)$
such that each $t$-subspace of $\cG_q(n,t)$ is contained in at least $\lambda$ codewords.
While $\cA_q (n,k,t;\lambda)$ represents the maximum size of a multiple packing,
the  quantity $\cC_q (n,k,t;\lambda)$ represents the minimum size of a multiple covering.
Clearly we have
\begin{theorem}
If $n$, $k$, and $t$ are positive integers such that $1 \leq t <k <n$ and $1 \leq \lambda \leq \sbinomq{n-t}{k-t}$, then
$$
\cA_q (n,k,t;\lambda) \leq \cC_q (n,k,t;\lambda) ~,
$$
with equality if and only if a subspace design $t$-$(n,k,\lambda)_q$ exists,
i.e. the bound of Theorem~\ref{thm:sphere} is attained with equality.
\end{theorem}

By Lemma~\ref{lem:t_k_complete}, each $t$-subspace of $\F_q^n$ is contained in $\sbinomq{n-t}{k-t}$ $k$-subspaces
of $\F_q^n$ and since the total number of $k$-subspaces of $\F_q^n$ is $\sbinomq{n}{k}$
we have a duality between a packing $\C$ of $k$-subspaces and the $k$-subspaces in
$\cG_q(n,k)$ which are not contained in $\C$. Hence, we have

\begin{theorem}
\label{thm:complement}
If $n$, $k$, and $t$ are positive integers such that $1 \leq t <k <n$ and $1 \leq \lambda \leq \sbinomq{n-t}{k-t}$, then
$$
\cA_q (n,k,t;\lambda) = \sbinomq{n}{k} - \cC_q \left(n,k,t;\sbinomq{n-t}{k-t} -\lambda \right) ~.
$$
\end{theorem}

\subsection{Complements}
\label{sec:complements}

By Corollary~\ref{cor:eq_perp2} the orthogonal complement code $\C^\perp$ of a $t$-$(n,k,\lambda)_q^m$ code $\C$ is
a $(\lambda+1)$-$(n,n-k,k-t+1)_q^c$ code. By Theorem~\ref{thm:complement}, the \emph{complement} code
of $\C$, $\C^c \deff \cG_q(n,k) \setminus \C$ is
a covering code with multiplicity $\sbinomq{n-t}{k-t} -\lambda$. If $\sbinomq{n-t}{k-t} -\lambda =1$ then $\C^c$ is a covering code,
i.e. a set of $k$-subspaces of~$\F_q^n$ such that each $t$-subspace of~$\F_q^n$ is contained
in at least one codeword of~$\C^c$.
It was proved in~\cite{EtVa11a} that $(\C^c)^\perp$ is a subspace $q$-Tur\'{a}n design, i.e. a set of $(n-k)$-subspaces of~$\F_q^n$
such that each $(n-t)$-subspace of $\F_q^n$ contains at least one codeword of~$(\C^c)^\perp$. This implies
various connections between all these types of codes. We would not go into detailed specific connections between codes
with given parameters. But, we will mention one specific code. The most interesting specific code which was heavily
studied is the $q$-analog of the Fano plane, i.e. the $q$-Fano plane~\cite{BEOVW,BKN15,EtHo18}. Such a code is a $2$-$(7,3,1)$ code
with $\sbinomq{7}{2} / \sbinomq{3}{2}$ codewords. It is not known whether such a code exists and there was
an extensive search and research for the largest code with these parameters. It was shown in~\cite{HKKW17} that for $q=2$ there exists
a code with such parameters and 333 codewords. The size of the smallest known related covering code is 396~\cite{Etz14}.
The size of the related $q$-Fano plane (if exists) is~381. These results coincide with the common belief that a construction of
a small covering code is simpler than a construction of related large packing code. In view of Theorem~\ref{thm:complement}
we have that $\cA_2 (7,3,2;30) \geq \sbinomtwo{7}{3} - 396$, while $\cC_2 (7,3,2;30) \leq \sbinomtwo{7}{3} - 333$.
This implies that the current known lower bound on the packing number is closer
to the packing bound (Theorem~\ref{thm:sphere}) than the related known upper bound on the
covering number. A~construction which might lead to a $q$-Fano plane for larger $q$ was recently given in~\cite{EtHo18}.
Related results on similar structures were proved recently in~\cite{BKKNW}.

\subsection{$\cA_q (n,k,t;\lambda)$ vs. $\cA_q (n,n-k,n-2k+t;\lambda)$}

The gap between the size of the largest code and the packing bound (Theorem~\ref{thm:sphere}) might be small
or large, depending on the parameters. Consider the case where
$n=6$, $k=4$, $t=3$, $\lambda =1$, and $q=2$. By Theorem~\ref{thm:sphere} we have that $\cA_2 (6,4,3;1) \leq 93$,
while the actual value is $\cA_2 (6,4,3;1) =21$. This is not unique for these parameters, and it occurs since
$\cA_q (n,k,t;1) =  \cA_q (n,n-k,n-2k+t;1)$, i.e. we should consider only $k \leq n/2$, when $\lambda =1$.
For $\lambda > 1$ there is no similar connection between codes of $\cG_q(n,k)$ and codes of $\cG_q(n,n-k)$.
Interesting phenomena might occur when $k > n/2$.
A good example can be given by considering $n=6$, $k=4$, $t=3$, and $\lambda=2$.
By Theorem~\ref{thm:sphere} we have that $\cA_2 (6,4,3;2) \leq 186$.
By Theorem~\ref{thm:Johnson1} we have that $\cA_2 (6,4,3;2) \leq \left\lfloor \frac{63}{15} \cA_2 (5,3,2;2) \right\rfloor$.
Hence, we have to consider the value of $\cA_2 (5,3,2;2)$. It is proved in~\cite{EKOO18} that
$\cA_2 (5,3,2;2) = 32$ and hence $\cA_2 (6,4,3;2) \leq 134$. This bound was improved
by a certain linear programming and a related construction was given,
so finally we have $121 \leq \cA_2 (6,4,3;2) \leq 126$~\cite{EKOO18}.
This is an indication that when $2k > n$ and $\lambda > 1$ the value of
$\cA_2 (n,k,t;\lambda >1)$ is relatively much larger than the one of
$\cA_2 (n,n-k,n-2k+t;\lambda > 1)$ (A comprehensive work on bounds and constructions
related to $\cA_q (n,k,t;\lambda)$ is discussed in~\cite{EKOO18}).

\subsection{Subspace Designs}

There is a sequence of subspace designs in which each $t$-subspace is contained in exactly $\lambda$ $k$-subspaces.
These designs form optimal generalized Grassmannian codes. These subspace designs
have small block length and very large $\lambda$. Hence, they are restricted  for solutions of generalized combination networks
for which the subspaces have dimension close to the one of the ambient space.
But, of course these subspace designs are optimal as Grassmannian codes. For example,
Thomas~\cite{Tho87} has constructed a $2$-$(n,3,7)_2$ design, where $n \geq 7$ and
$n \equiv 1$ or $5~(\text{mod}~6)$. This is the same as a $2$-$(n,3,7)^m_2$ code and its orthogonal
complement is an\linebreak $8$-$(n,n-3,2)^c_2$ code. By Theorem~\ref{thm:gen_RiAh},
such a code provides a solution for the
$(1,n-3)$-$\mathcal{N}_{n,(2^n-1)(2^{n-1}-1)/3,8(n-3)+1}$ network.
Some interesting codes and related networks exist for small parameters.
For example, consider the $2\text{-}(6,3,3)_2$ design presented in~\cite{BKL05},
which is a $2$-$(6,3,3)^m_2$ code and its orthogonal complement is a $4$-$(6,3,2)^c_2$ code.
By Theorem~\ref{thm:gen_RiAh}, such a code provides a solution for the $(1,3)$-$\mathcal{N}_{6,279,13}$ network.
Other designs which lead to optimal codes for related networks can be found in many recent
papers on this emerging topic, e.g.~\cite{BEOVW,BKL05,BKN15,BKKNW,Etz18,EtHo18,FLV14,KiLa15,KiPa15,MMY95,RaSi89,Suz90,Suz90a,Suz92,Tho87,Tho96}.

\subsection{Network Coding Solutions}

How do the results affect the gaps between vector network coding and scalar linear network coding?
A good example can be given by considering the ${(1,1)\text{-}\mathcal{N}_{3,r,4}}$ network
which was used in~\cite{EtWa15} to show that there is a network with three messages for which vector network coding
outperforms scalar linear network coding.

By Theorem~\ref{thm:gen_RiAh_Vec}, for a vector solution of this network a $3$-$(3 \ell,\ell,\ell)^c_q$ code is required.
For simplicity and for explaining the problems in obtaining lower and upper bounds on $\cA_q(n,k,t;\lambda)$ we consider
the case of $t=3$, $\ell=2$, and $q=2$, i.e. a $3$-$(6,2,2)^c_2$ code or its orthogonal complement,
a $3$-$(6,4,2)^m_2$ code. In~\cite{EtWa16} a code with 51 codewords was presented.
When a $3$-$(6,2,2)^c_2$ code was considered a code of size 121
(related to the bound $121 \leq \cA_2 (6,4,3;2) \leq 126$ mentioned before) was obtained~\cite{EKOO18}.

For a related scalar linear coding solution the related alphabet size is 4 (since
for the vector coding binary vectors of length 2 were considered). By Theorem~\ref{thm:gen_RiAh},
for a scalar linear network coding solution we need a $3\text{-}(3,1,1)_4^c$ code. The largest such code consists of
the 21 one-dimensional subspaces of $\F_4^3$, each one is contained twice
in the code. Therefore, the number of nodes in the middle layer
can be at most 42, while for vector network coding 121 nodes in the middle layer can be used (since $121 \leq \cA_2 (6,4,3;2)$).
The smallest alphabet size for which  a scalar solution with 121 nodes in the middle layer exists
is 8.  There are 73 one-dimensional subspaces of $\F_8^3$ and each one can be used twice in the
code, but only 121 codewords out of these 146 possible codewords are required.
This is the indication on the superiority of vector
network coding on scalar linear network coding.

\section{Hamming Codes vs. Grassmannian Codes}
\label{sec:approach_Hamm}

In this section we will examine codes in the Hamming scheme related to the
new distance measures. This will be done in two different directions.
In Section~\ref{sec:subfam_Hamm} sets of codes in the Hamming scheme will be considered as
subfamilies of the Grassmannian codes. In Section~\ref{sec:Hamm} we will consider
the related distance measures in the Hamming scheme. Finally, in Section~\ref{sec:gen_weight}
we show how the concepts of generalized Hamming weights can be used in the Grassmannian scheme
using the new distance measures.

\subsection{Hamming Codes as a Subfamily of Grassmannian Codes}
\label{sec:subfam_Hamm}

It should be noted that codes in the Hamming space with the Hamming distance form
a subfamily of the Grassmannian codes. This can be shown using a few different approaches.
The first one was pointed out in 1957 by Tits~\cite{Tits57} who suggested that combinatorics of sets could
be regarded as the limiting case $q \to 1$ of combinatorics of vector spaces
over the finite field~$\F_q$. This can be viewed as the traditional approach,
but it seems that this is not the approach that we need for our purpose.
The goal is to show that Hamming codes form a subfamily of the Grassmannian codes
based on network coding solutions for the generalized combination networks and the new distance measures.

The first approach that we suggest is based on the network coding solutions
for the combination networks.
By Theorem~\ref{thm:RiAh}, the $(0,1)$-$\mathcal{N}_{h,r,\alpha}$ network has a scalar solution
over~$\F_q$, if and only if there exists a code $\cC$, over~$\F_q$, of length $r$, $q^h$ codewords, and minimum Hamming
distance $r-\alpha+1$. This code can be either linear or nonlinear. If the code is nonlinear then some of the
coding functions on the $r$ edges between the first layer and the second layer in the
$(0,1)$-$\mathcal{N}_{h,r,\alpha}$ network are nonlinear and our framework using
$\alpha$-$(h,k,\delta)^c_q$ codes is not the appropriate one. Such nonlinear codes
for nonlinear network coding will be considered in the next paragraph.
If the code $\cC$ is linear then our framework is indeed a $q$-analog for the solution with $\cC$.
For such a code $\cC$ we have an $h \times r$ generator matrix~$G$ for which any set of $\alpha$ columns
from~$G$ has a subset of $h$ linearly independent columns. The $r$ columns, considered
as one-subspaces of $G$, form an $\alpha$-$(h,1,h-1)^c_q$ code
as implied by Theorem~\ref{thm:gen_RiAh}. The number of codewords in the largest such a code is
the same as the largest length~$r$ of a code $\cC$ with dimension~$h$ and minimum Hamming distance $r-\alpha+1$.
We note that if $h$ and $\alpha$ are fixed then the minimum Hamming distance depends on this largest length $r$.
Hence, it is more natural to search for the largest $\alpha$-$(h,1,h-1)^c_q$ code for which the minimum
Hamming distance is not a parameter, when the $r$ codewords are taken as the columns of an $h \times r$ generator matrix.

In the second approach, a solution for the $(\epsilon,k)$-$\mathcal{N}_{h,r,\alpha k+\epsilon}$ network
based on a code in the Hamming scheme is considered.
In other words, let us consider a scalar solution for the network based on a code over $\F_q$ in the Hamming scheme rather than
a solution based on a code in the Grassmannian $\cG_q(h,k)$ as was done in the
previous sections. The code can be linear or nonlinear and we distinguish between
these two cases. There will be another distinction in the formulation of the network
coding solutions. For a code in the Grassmannian scheme,
the coefficients (which form the local coding vectors) of the linear functions on the edges of the network are the entries of the vectors
related to a basis of the subspaces from the Grassmannian code. If the $k$ columns of each such basis were written in an
$h \times (rk)$ matrix $G$ (each node in the middle layer yields $k$ consecutive columns for this matrix),
then~$G$ would have been the generator matrix of a code in the Hamming scheme which
form a solution for the network. In fact an optimal solution (maximum number of nodes in the middle layer)
with a code in the Grassmannian scheme
will induce an optimal linear solution (largest length) with a code in the Hamming scheme and vice versa.
But, we are interested in analysing the solution for the network using
a code in the Hamming scheme as this might give us extra information and maybe codes with smaller
alphabet size (clearly only in the case of nonlinear codes).
Such a code (nonlinear over $\F_q$) must have length $rk$ as for the linear code. The source transmit $k$ consecutive
symbols of a codeword to each different
node from the middle layer, where the $i$-th node will receive the $i$-th set
of $k$ consecutive symbols of the codeword. The code should have $q^h$
codewords since the source must send a codeword to the nodes in the middle layer for each one of the different
$q^h$ messages. This is the point to make it clear that some codewords sent by the source might be identical. This would not cause
any problem since identical codewords can be distinguished at the receivers by the information sent from
the source to the receiver through the $\epsilon$ direct links. For simplicity the reader can assume first
that $k=1$, although the description will be for any
$k \geq 1$. Each sub-codeword of length~$\alpha k$ in a projection of $\alpha k$ coordinates,
related to $\alpha$ nodes in the middle layer, appears at most $q^\epsilon$
times as such a sub-codeword in the projection of these $\alpha k$ coordinates.
In this way, the source can send $\epsilon$ symbols, to the related
receiver, that will distinguish between the codewords which have the same values in these $\alpha k$
coordinates. Such codes which have larger values of $r$ than the linear ones exist.
For example, the vector network codes
which outperform the scalar linear network codes can be translated to scalar nonlinear
network codes. But, to find in this way nonlinear codes with
larger length than the ones obtained by the vector network codes is an interesting and intriguing question
for future research. When we consider a linear code the last requirement is slightly stronger (and weaker
in the sense of obtaining larger codes).
Each subset of $\alpha$ nodes from the middle layer have in the projection
of their $\alpha k$ coordinates, in the codeword, at least $h-\epsilon$ linearly independent vectors. This corresponds
exactly to the fact that each subset of $\alpha$ codewords in the Grassmannian code spans a subspace of $\F_q^h$ whose
dimension is at least $h- \epsilon$. This analysis simulates exactly the requirements from a scalar linear solution
for the $(\epsilon,k)$-$\mathcal{N}_{h,r,\alpha k+\epsilon}$ network.
To end this analysis we have to explain what are the nonlinear functions written on the edges.
For each $h$ messages and any given $k$ links from the source to a node $v$ in the middle layer,
the source should send a symbol from $\F_q$. This clearly define a function with $h$ variables in $\F_q$.
Each one of the $q^h$ different combinations of the $h$ messages implies $k$ symbols
to the $k$ edges from the source to the $i$th nodes in the middle layer.
This implies some functions of the $h$ messages. These functions are written on the $k$ links which
connect the source to the $i$th node of the middle layer.

Linear codes in the Hamming scheme can be used also in other ways as solutions for the
generalized combination network. Let $\cC$ be a linear code of length $r$, dimension $k$,
and minimum Hamming distance $d$, over $\F_q$. For such a code, there exists an $(r-k) \times r$ parity-check matrix~$H$.
In~$H$, each $d-1$ columns are linearly independent. Hence, the $r$ columns of~$H$ form a $(d-1)$-$(r-k,1,d-2)^c_q$ code.
By Theorem~\ref{thm:gen_RiAh}, this code solves the $(r-k-d+1,1)$-$\mathcal{N}_{r-k,r,r-k}$
network. The number of codewords in the largest such a Grassmannian code
is the same as the largest length~$r$ of a code of dimension $k$, and minimum Hamming distance $d$. Also, in this
case we note that in general we would like to obtain the largest $r$ when $r-k$ and $d$ are fixed.
Hence, it is more natural to search for the largest $(d-1)$-$(h,1,d-2)^c_q$ code for which $d$ and $h$ are fixed.
There are related results from arcs in projective planes. Since these
results are limited for a space of dimension three we omit these results. The interested
reader is referred to~\cite{HiSt01}.

The conclusion from this discussion is that the codes in the Hamming scheme can be viewed as a
subfamily of Grassmmannian codes. But, it is still interesting to find nonlinear codes
which outperform vector codes as solutions for the generalized combination networks.

\subsection{Related Codes in the Hamming Scheme}
\label{sec:Hamm}

We have defined two new distance measures and their related codes in the Grassmann space.
What are the related codes in the Hamming scheme (or more precisely in the Johnson scheme)?
This family of codes was defined before, but usually was considered for limited number of parameters.
A $t$-$(n,k,\lambda)$ packing for $1 \leq t \leq k \leq n$, is a collection of $k$-subsets (called \emph{blocks})
from an $n$-set $Q$, such that each $t$-subset of $Q$ is contained in at most
$\lambda$ blocks. These packings were considered for small $t$ and $k$ and mainly
optimal packings were considered. For the known results on
this topic the reader is referred to~\cite{MiMu92,SWY06} and references therein.
Generally, each such packing can be viewed as a code
of length $n$. If the $n$-set is $\{ 1,2,\ldots,n\}$ then a $k$-subset $X$ is
translated into a codeword of length $n$ and weight $k$, where the \emph{ones}
are in the coordinates indexed by the $k$-subset $X$.

Let $A(n,k,t;\lambda)$ be the maximum number of $k$-subsets in a $t$-$(n,k,\lambda)$ packing.
For these new families of codes only codes with no repeated codewords are considered.
When $k > n/2$ the bounds on the sizes of such codes are related to the Tur\'{a}n's problem~\cite{Kee11},
in a similar way to the complements in Section~\ref{sec:complements}.

For example we consider the value of $A(n,n-2,n-3;\lambda)$. An $(n-3)$-subset of~$[n]$
is contained in exactly three $(n-2)$-subsets of $[n]$. Hence, $\lambda$ can have only the values
$\lambda=1$, $\lambda=2$, or $\lambda=3$. If $\lambda =1$ then the value is just the packing number
$A(n,n-2,n-3;1)=A(n,2,1;1)= \left\lfloor \frac{n}{2} \right\rfloor$. If $\lambda=3$ then
there exists a trivial code which contains all the $(n-2)$-subsets of $[n]$. This code
attains the value $A(n,n-2,n-3;3)= \binom{n}{2}$. The only nontrivial value to consider
is $A(n,n-2,n-3;2)$ which is a well-known Tur\'{a}n number. Hence, we have
for $n \geq 2$, $A(2n,2n-2,2n-3;2)=n^2$ and for $n \geq 2$, $A(2n+1,2n-1,2n-2;2)=n(n+1)$.
Note, that by the packing bound we have $A(n,n-2,n-3;1) \leq \frac{n(n-1)}{6}$,
$A(n,n-2,n-3;2) \leq \frac{n(n-1)}{3}$, and $A(n,n-2,n-3;3) \leq \binom{n}{2}$.
While $A(n,n-2,n-3;3)$ attains the packing bound, the
value of $A(n,n-2,n-3;1)$ is far from the bound. As for  $A(n,n-2,n-3;2)$ we have an asymptotic ratio
of $3/4$ between the exact values and the upper bound of the packing bound.
The behavior, of these ratios, when there
are many more possible values for $\lambda$ is an interesting question.

\subsection{Generalized Weights}
\label{sec:gen_weight}

The definition of $\alpha$-Hamming covering codes (which is the related
definition for \linebreak$\alpha$-{Grassmannian covering codes) for constant weight codes can be straightforward generalized
to any code in the Hamming scheme, where the $\alpha$-Hamming covering of $\alpha$ nonzero codewords in a code $\cC$
is the number of coordinates which are nonzero in some of the $\alpha$ nonzero codewords. The minimum
$\alpha$-Hamming covering of a code $\cC$ is the minimum $\alpha$-Hamming covering on any set of $\alpha$ distinct nonzero codewords of $\cC$.
We can form a hierarchy of the $\alpha$-Hamming covering for $\cC$, from $\alpha=1$, $\alpha=2$, up to $\alpha=|\cC|$.
This hierarchy will be denoted by $c_1,c_2,\ldots,c_{|\cC|}$.
This hierarchy, for the Grassmannian codes, can be viewed as
a $q$-analog for the generalized Hamming weights~\cite{Wei91}.
Previously a geometric approach for these generalized weights was given in~\cite{TsVl95}.
Let~$\cC$ be a linear code of length $n$ and dimension $k$ and $\cA$ be a linear subcode of~$\cC$.
The \emph{support} of $\cA$, denoted  by $\chi(\cA)$, is defined by
$$
\chi(\cA)\deff\{i:\exists(a_1,a_2,\ldots,a_n)=\textbf{a}\in \cA,~a_i\neq 0\}.
$$

The $r$th \emph{generalized Hamming weight} of a linear code~$\cC$, denoted by  $d_r(\cC)$ ($d_r$ in short), is
the minimum support of any $r$-dimensional subcode of~$\cC$, $1\leq r\leq k$, namely,
$$
d_r=d_r(\cC)\deff \min_{\cA}\{|\chi(\cA)|:
$$
$$
\cA \text{~is~a~linear~subcode~of}~\cC, \dim (\cA) = r\}.
$$
Clearly, $d_r\leq d_{r+1}$ for $1\leq r\leq k-1$.
The set $\{d_1,d_2,\ldots,d_k\}$ is called the \emph{generalized Hamming
weight hierarchy} of~$\cC$~\cite{Wei91}.

It is not difficult to verify that for a binary linear code $\cC$ we have,
$c_1 = d_1$, $c_2=c_3=d_2$, $c_4=c_5=c_6=c_7=d_3$, and in general
$c_{2^i} = c_{2^i+1} = \cdots = c_{2^{i+1}-1}=d_{i+1}$.
This direction of research is interesting, especially when we consider the $q$-analog
of the generalized Hamming weights in
the Grassmann scheme and a straightforward generalization to subspace codes which are not
necessarily constant dimension. This possible $q$-analog for the generalized Hamming weights
can be of great interest and it is a topic for future research.

%%%%%%%%%%%%%%%%%%%%%%%%%%%%%%%%%%
%%%
%%%       CONCLUSION
%%%
%%%%%%%%%%%%%%%%%%%%%%%%%%%%%%%%%%%%

%\vspace{-0.1cm}

\section{Conclusions and Open Problems}
\label{sec:conclude}

We have introduced a new family of Grassmannian codes
with two new distance measures which generalize the traditional
Grassmannian codes with the Grassmannian distance. There is a correspondence
between the set of these codes of maximum size and the optimal scalar linear solution
for the set of generalized combination networks.
Based on the generalized combination networks and the new distance measures we have
proved that codes in the Hamming scheme can be viewed as a subfamily of the Grassmannian codes
from a few different angles.
The research we have started for bounds on the sizes of such codes is very preliminary
and there are many obvious coding questions related to these codes, some of which are currently under research~\cite{EKOO18}
and will provide lots of ground for future research. Other interesting problems were suggested throughout our exposition
and are summarized as follows.
\begin{enumerate}
\item Find scalar nonlinear network codes which outperform vector network codes
on generalized combination networks. What is the maximum gap in the alphabet size
and/or the number of nodes in the middle layer of the generalized combination networks related to these codes?

\item Consider the hierarchy for the $q$-analog of the generalized Hamming weights.

\item Find more applications for the new classes of codes in network coding.
\end{enumerate}

%\vspace{-0.1cm}

\section*{Acknowledgement}

The authors thank Jerry Anderson Pinheiro and Yiwei Zhang for pointing~\cite{Kee11}
and the connections with Tur\'{a}n's numbers. They also thank the associate editor and the reviewers for their
constructive comments.

%%%%%%%%%%%%%%%%%%%%%%%%%%%%%%%%%%%%%%%%%%%%%%%%%%%%%%%%%%%%%%%%%%%%%%
%%%%%%%%%%%%%%%%%%%%%%%%%%%%%%%%%%%%%%%%%%%%%%%%%%%%%%%%%%%%%%%%%%%%%%
%%%%%%%%%%%%%%%%%%%%%%%%%%%%%%%%%%%%%%%%%%%%%%%%%%%%%%%%%%%%%%%%%%%%%%
%%%%%%%%%%%%%%%%%%%%%%%%%%%%%%%%%%%%%%%%%%%%%%%%%%%%%%%%%%%%%%%%%%%%%%

\begin{IEEEbiographynophoto}{Tuvi Etzion}
(M'89--SM'94--F'04) was born in Tel Aviv, Israel, in 1956.  He
received the B.A., M.Sc., and D.Sc.~degrees from the Technion --
Israel Institute of Technology, Haifa, Israel, in 1980, 1982, and
1984, respectively.  From 1984 he held a position in the Department
of Computer Science at the Technion, where he now holds the Bernard
Elkin Chair in Computer Science. During the years 1985--1987 he was
Visiting Research Professor with the Department of Electrical
Engineering -- Systems at the University of Southern California, Los
Angeles. During the summers of 1990 and 1991 he was visiting
Bellcore in Morristown, New Jersey. During the years 1994--1996 he
was a Visiting Research Fellow in the Computer Science Department at
Royal Holloway University of London, Egham, England. He also had
several visits to the Coordinated Science Laboratory at University
of Illinois in Urbana-Champaign during the years 1995--1998, two
visits to HP Bristol during the summers of 1996, 2000, a few visits
to the Department of Electrical Engineering, University of
California at San Diego during the years 2000--2017, several
visits to the Mathematics Department at Royal Holloway University of
London, during the years 2007--2017, and a few visits to the
School of Physical and Mathematical Science (SPMS), Nanyang Technological University, Singapore,
during the years 2016--2019.

His research interests include applications of discrete mathematics
to problems in computer science and information theory, coding
theory, network coding, and combinatorial designs.

Dr.~Etzion was an Associate Editor for Coding Theory for the \textsc{IEEE
Transactions on Information Theory} from 2006 till 2009. From 2004 to
2009, he was an Editor for the \emph{Journal of Combinatorial Designs}.
From 2011 he is an Editor for Designs, Codes, and Cryptography, and
from 2013 an Editor for Advances of Mathematics in Communications.
\end{IEEEbiographynophoto}

\begin{IEEEbiographynophoto}{Hui Zhang}
received the Ph.D. degree in Applied Mathematics from Zhejiang
University, Hangzhou, Zhejiang, P. R. China in 2013. During 2012--2015, she
used to work as a Project Officer and a Research Fellow at the School of
Physical and Mathematical Sciences at Nanyang Technological University,
Singapore. During 2015--2017, she was a Postdoctoral Researcher at the
Computer Science Department at Technion - Israel Institute of Technology.
Currently, she is a Research Fellow at the School of Physical and
Mathematical Sciences at Nanyang Technological University, Singapore. Her
research interests include combinatorial theory, coding theory and
cryptography, and their intersections.
\end{IEEEbiographynophoto}


\begin{thebibliography}{99}
\bibitem{ACLY00}
    {R. Ahlswede, N. Cai, S.-Y. R. Li, and R. W. Yeung,}
    ``Network information flow'',
    {\em IEEE Trans.\ Inform. Theory}, vol.\,46, pp.\,1204--1216, 2000.
\bibitem{BlEt11}
    {S. Blackburn and T. Etzion,}
    ``The asymptotic behavior of Grassmannian codes'',
    {\em IEEE Trans.\ Inform. Theory}, vol.\,58, pp.\,6605--6609, 2012.
\bibitem{BEOVW}
     {M. Braun, T. Etzion, P. R. J. \"Osterg\aa rd, A. Vardy, and A. Wassermann,}
     ``Existence of $q$-analogs of Steiner systems,''
     {\em Forum of Mathematics, Pi}, vol. 4, pp. 1--14, 2016.
\bibitem{BKL05}
    {M. Braun, A. Kerber, and R. Laue,}
    ``Systematic construction of $q$-analogs of $t-(v,k,\lambda)$-designs'',
    {\em Designs, Codes, and Cryptography,} vol.\,34, pp.\,55--70, 2005.
\bibitem{BKN15}
     {M. Braun, M. Kiermaier, and A. Naki\'{c},}
     ``On the automorphism group of a binary $q$-analog of the Fano plane,''
    {\em European Journal of Combinatorics,} vol. 51, pp. 443--457, 2016.
\bibitem{BKKNW}
     {M, Buratti, M. Kiermaier, S. Kurz, A. Naki\'{c}, and A. Wassermann}
      ``$q$-analogs of group divisible designs,''
      {\em arxiv.org/abs/1611.06827}, April 2018.
\bibitem{EbFr11}
   {J. B. Ebrahimi and C. Fragouli,}
   ``Algebraic algorithms for vector network coding'',
   {\em IEEE Trans.\ Inform. Theory}, vol.\,57, pp.\,996--1007, 2011.
\bibitem{Etz14}
    {T. Etzion,}
    ``Covering of subspaces by subspaces,''
    {\em Designs, Codes, and Cryptography,} vol. 72, pp. 405--421, 2014.
\bibitem{Etz18}
    {T. Etzion,}
    ``A new approach for examining $q$-Steiner systems,''
     {\em The Electronic Journal of Combinatorics,}
     vol. 25, \#P2.8, 2018.
\bibitem{EGRW16} T. Etzion, E. Gorla, A. Ravagnani, and A. Wachter-Zeh,
     ``Optimal Ferrers diagram rank-metric codes,''
     {\it IEEE Trans.\ Inform.\ Theory,}
     vol. 62, pp. 1616--1630, 2016.
\bibitem{EtHo18}
     {\sc T. Etzion and N. Hooker,}
     ``Residual $q$-Fano plane and related structures,''
     {\em The Electronic Journal of Combinatorics,}
     vol. 25, \#P2.3, 2018.
\bibitem{EKOO18}
    {T. Etzion, S. Kurz, K. Otal, and F. \"{O}zbudak,}
   ``Subspace packings'',
     {\em arxiv.org/abs/1811.04611}, November 2018.
\bibitem{EtSi09}
    {T. Etzion and N. Silberstein,}
    ``Error-correcting codes in projective
         space via rank-metric codes and Ferrers diagrams,''
    {\em IEEE Trans. Inform. Theory,} vol. 55, pp. 2909--2919, 2009.
\bibitem{EtSi13}
    {T. Etzion and N. Silberstein,}
    ``Codes and designs related to lifted MRD codes'',
    {\em IEEE Trans.\ Inform. Theory,} vol.\,59, pp. 1004--1017, 2013.
\bibitem{EtSt16}
  {T. Etzion and L. Storme,}
  ``Galois geometries and coding theory,''
  {\em Designs, Codes, and Cryptography,} vol. 78, pp. 311--350, 2016.
\bibitem{EtVa11}
     {T. Etzion and A. Vardy,}
     ``Error-correcting codes in projective space'',
     {\em IEEE Trans. on Inform. Theory,} vol.\,57, pp. 1165--1173, 2011.
\bibitem{EtVa11a}
     {T. Etzion and A. Vardy,}
     ``On $q$-analogs for Steiner systems and covering designs,''
     {\em Advances in Mathematics of Communications,} vol. 5, pp. 161--176, 2011.
\bibitem{EtWa15}
    {T. Etzion and A. Wachter-Zeh,}
    ``Vector network coding based on subspace codes outperforms scalar linear network coding'',
    {\em IEEE Trans. on Inform. Theory,}  vol.\,64, pp.\,2460--2473, 2018.
\bibitem{EtWa16}
    {T. Etzion and A. Wachter-Zeh,}
    ``Vector network coding based on subspace codes outperforms scalar linear network coding'',
    {\em Proc. of IEEE Int. Symp. on Inform. Theory (ISIT)}, pp.\,1949--1953, Barcelona, Spain, July 2016.
\bibitem{FLV14}
    {A. Fazeli, S. Lovett, and A. Vardy,}
    ``Nontrivial $t$-designs over finite fields exist for all $t$'',
    {\em J.\ Combin.\ Theory Ser.\,A,} vol. 127, pp. 149--160, 2014.
\bibitem{FrSo16}
    {C. Fragouli and E. Soljanin,}
    ``(Secure) Linear network coding multicast'',
    {\em Designs, Codes, and Crypt.}, vol.\,78, pp.\,269--310, 2016.
\bibitem{HKKW17}
     {D. Heinlein, M. Kiermaier, S. Kurz, A. Wassermann,}
     ``A subspace code of size 333 in the setting of a binary q-analog of the Fano plane,''
     {\em arxiv.org/abs/1708.06224}, December 2017.
\bibitem{HiSt01}
    {J. W. P. Hirschfeld and L. Storme,}
    ``The packing problem in statistics, coding theory and ﬁnite projective spaces: Update 2001'',
     {\em Finite Geometries (Developments in Mathematics)}, vol.\,3. Dordrecht, The Netherlands: Kluwer, pp.\,201--246, 2001.
\bibitem{LYC03}
    {S.-Y. R. Li, R. W. Yeung, and N. Cai,}
    ``Linear network coding'',
    {\em IEEE Trans.\ Inform. Theory}, vol.\,49, pp.\,371--381, 2003.
\bibitem{Kee11}
     {P. Keevash,}
     ``Hypergraph Tur\'{a}n problems,''
     in {\em Surveys in combinatorics}, Cambridge University Press, pp. 83--140, 2011.
\bibitem{KiKu17}
    {M.~Kiermaier and S.~Kurz,}
    ``An improvement of the Johnson bound for subspace codes,''
    {\em arXiv preprint 1707.00650}, 2017.
\bibitem{KiLa15}
     {M. Kiermaier and R. Laue,}
     ``Derived and residual subspace designs,''
     {\em Advances in Mathematics of Communications}, vol. 9, pp. 105--115, 2015.
\bibitem{KiPa15}
     {M. Kiermaier and M. O. Pav\v{c}evi\'{c},}
     ``Intersection numbers for subspace designs,''
     {\em Journal of Combinatorial Designs}, vol. 23, pp. 463--480, 2015.
\bibitem{KoMe03}
    {R. K\"{o}tter and M. M\'{e}drad,}
    ``An algebraic approach to network coding'',
    {\em IEEE Trans.\ Networking}, vol.\,11, pp.\,782--795, 2003.
\bibitem{KoKs08}
    {R. K\"{o}tter and F. R. Kschischang,}
    ``Coding for errors and erasures in random network coding'',
    {\em IEEE Trans. on Inform. Theory}, vol.\,54, pp. 3579--3591, 2008.
\bibitem{MiMu92}
     {W. H. Mills and R. C. Mullin,}
     {\sl Coverings and packings}
     in {\em Contemporary Design Theory: A Collection of Surveys},
     edited by {\sc J. H. Dinitz and D. R. Stinson,}
     {\sl John Wiley: New York, 1992}.
\bibitem{MMY95}
    {M. Miyakawa, A. Munemasa and S. Yoshiara,}
    ``On a class of small\/ $2$-designs over ${\rm GF}(q)$,''
    {\em Journal Combinatorial Designs}, vol. 3, pp. 61--77, 1995.
\bibitem{RaSi89}
    {D. K. Ray-Chaudhuri and N. M. Singhi},
    ``$q$-analogues of $t$-designs and their existence,''
    {\em Linear Algebra Appl.}, vol. 114/115, pp. 57--68, 1989.
\bibitem{RiAh06}
     {S. Riis and R. Ahlswede,}
     ``Problems in network coding and error correcting codes''
         appended by a draft version of S. Riis
       ``Utilising public information in network coding'',
      {\it General Theory of Information Transfer and Combinatorics,}
     in {\it Lecture Notes in Computer Science,}
     vol.\,4123 pp.\,861--897, 2006.
\bibitem{SiTr15}
     {N. Silberstein and A.-L. Trautmann,}
     ``Subspace codes based on graph matchings, Ferrers diagrams, and pending blocks'',
     {\em IEEE Trans. Inform. Theory}, vol. 61, pp. 3947--3953, 2015.
\bibitem{SKK08}
      {D.~Silva, F.~R. Kschischang, and R.~K{\"o}tter,}
       ``A rank-metric approach to error control in random network coding,''
      {\em IEEE Trans. Inform. Theory}, vol.~54, pp. 3951--3967, 2008.
\bibitem{SWY06}
     {D. R. Stinson, R. Wei, and J. Yin,}
     {\sl Packings}
     in {\em Handbook of Combinatorial Designs},
     edited by {\sc C. J. Colbourn and J. H. Dinitz}
     {\sl Chapman \& Hall/CRC, Boca Raton, 2006.}.
\bibitem{Suz90}
    {H. Suzuki,}
    ``$2$-designs over ${\rm GF}(2^m)$,''
    {\em Graphs Combin.}, vol. 6, pp. 293--296, 1990.
\bibitem{Suz90a}
    {H. Suzuki,}
    ``On the inequalities of $t$-designs over a finite field,''
    {\em European Journal Combin.}, vol. 11, pp. 601--607, 1990.
\bibitem{Suz92}
    {H. Suzuki,}
    ``$2$-designs over ${\rm GF}(q)$,''
    {\em Graphs Combin.}, vol. 8, pp. 381--389, 1992.
\bibitem{Tho87}
    {S. Thomas,}
    ``Designs over finite fields'',
    {\em Geometriae Dedicata,} vol.\,21, pp.\,237--242, 1987.
\bibitem{Tho96}
    {S. Thomas,}
    ``Designs and partial geometries over finite fields,''
    {\em Geometriae Dedicata,} 63 (1996), 247--253.
\bibitem{Tits57}
    {J. Tits},
    {\sl Sur les analogues alg\'ebriques des groupes semi-simples complexes},
    in {\em Colloque d'Alg\'ebre Sup\`erieure},
    tenu \'a Bruxelles du 19 au 22 d\'ecembre 1956,
    Centre Belge de Recherches Math\'ematiques \'Etablissements Ceuterick,
    Louvain, Paris: Librairie Gauthier-Villars, pp.\,261--289, 1957.
\bibitem{TsVl95}
    {M. A. Tsfasman and S. G. Vl\v{a}du\c{t},}
    ``Geometric approach to higher weights'',
    {\em IEEE Trans. on Inform. Theory}, vol.\,41, pp. 1564--1588, 1995.
\bibitem{vLWi92}
    {\sc J. H. van Lint and R. M. Wilson,}
    {\em A course in Combinatorics},
    Cambridge University Press, 1992.
\bibitem{Wei91}
    {V. K. Wei,}
    ``Generalized Hamming weights for linear codes,''
    {\em IEEE Trans.\ Inform. Theory}, vol.\,37, pp.\,1412--1418, 1991.
\bibitem{XiFu09}
     {S.-T. Xia and F.-W. Fu,}
     ``Johnson type bounds on constant dimension codes,''
     {\em Designs, Codes, and Cryptography,} vol. 50, pp. 163--172, 2009.
\end{thebibliography}
\end{document}